\documentclass[a4paper,12pt]{article} 
\usepackage[english]{babel}   
\usepackage[utf8]{inputenc} 
\usepackage[babel]{csquotes}
\usepackage[T1]{fontenc} 
\usepackage{lmodern}     
\usepackage{microtype}
\usepackage{comment}
\usepackage{enumitem}
\usepackage{nicefrac, xfrac}
\usepackage{setspace}
\usepackage{geometry}
\usepackage[natbibapa]{apacite}
\usepackage{amsmath,amssymb,amsthm,amsfonts,array,mathrsfs,ifthen,mathtools}
\usepackage{dsfont}
\usepackage[dvipsnames]{xcolor}
\definecolor{coolblack}{rgb}{0.0, 0.18, 0.39}
\usepackage[colorlinks=true,linkcolor=coolblack,urlcolor=coolblack,citecolor=coolblack,pdftex]{hyperref}
\usepackage[colorinlistoftodos]{todonotes}
\newcommand{\FN}[1]{\todo[color=cyan!50!white!50,size=\tiny]{#1}}
\usepackage{graphicx,subcaption}
\usepackage{booktabs} 
\usepackage{caption}                      
\usepackage{quoting} 
\usepackage{soul} 
\usepackage{tikz}                          
\usetikzlibrary{arrows,decorations,backgrounds,calc,shapes}

\DeclareMathOperator*{\argmax}{\arg\!\max}

\theoremstyle{plain}
\newtheorem*{theorem*}{Theorem}

\newtheorem{definition}{Definition}
\newtheorem{proposition}{Proposition}
\newtheorem{lemma}{Lemma}
\newtheorem{corollary}{Corollary}
\newtheorem{step}{Step}

\title{Belief patterns with information processing\thanks{I thank Ennio Bilancini, Santiago Oliveros, Ludovic Renou, and Katharine Rockett for helpful comments and suggestions on a much earlier draft. All errors are mine.}}

\author{Federico Vaccari\thanks{Department of Economics, University of Bergamo, \textit{e-mail}: \href{mailto:vaccari.econ@gmail.com}{\sf vaccari.econ@gmail.com}.}}

\date{}		

\begin{document}



\maketitle 	

\begin{abstract}
\noindent This paper presents a model of costly information acquisition where decision-makers can choose whether to elaborate information superficially or precisely. The former action is costless, while the latter entails a processing cost. Within this framework, decision-makers' beliefs may polarize even after they have access to the same evidence. From the perspective of a Bayesian observer who neglects information processing constraints, the decision-makers' optimal behavior and belief updating may appear consistent with biases such as disconfirmation, underreaction to information, and confirmation bias. However, these phenomena emerge naturally within the model and are fully compatible with standard Bayesian inference and rational decision-making when accounting for the costs of information acquisition.
\end{abstract}

	
	
	\thispagestyle{empty}

\newpage

{\footnotesize

\tableofcontents
}

\section{Introduction}

Individuals with access to the same information can arrive at increasingly divergent beliefs or update their opinions in ways that fail to reflect the available evidence. This observation is consistent with a substantial body of literature in economics and psychology that documents systematic biases in how information is processed. Understanding how individuals incorporate new information is vital, as belief polarization can drive financial frenzies, escalate conflicts, and deepen economic, political, and cultural divisions. This paper employs a fully rational Bayesian framework to examine decision-makers' incentives for information acquisition and their implications on belief updating.  

There are three key observations that define the approach taken here. First, having access to information does not necessarily mean fully incorporating its content---just as owning an encyclopedia does not equate to knowing all of its entries. Second, scrutinizing evidence with greater precision is possible but comes at a cost. Third, individuals with different beliefs have different incentives for examining information. This is because the value of information lies in its potential to influence optimal decision-making. For example, staunch supporters of a political party may have little interest in carefully analyzing the platforms or debates of opposing parties, as their choice is already decided. By contrast, undecided voters are more motivated to seek political information that helps them make an informed decision when casting their ballot.

This paper examines a parsimonious model of information acquisition where rational and Bayesian decision-makers must determine which of two states is true. They have access to an imperfect signal that provides information about the true state, and they can choose to scrutinize it either superficially or with greater precision. Superficial scrutiny comes at no cost, while precise acquisition requires decision-makers to incur a processing cost. Their trade-off is between making a less informed decision for free or incurring a cost to make a better-informed choice.

Processing costs take the form of a disutility required to obtain additional information, encapsulating concepts developed in both economics and psychology. In their simplest form, they may represent a direct economic cost, such as a fee to consult an expert or the price of accessing a database. More broadly, they encompass the resources needed to process larger volumes of information, including time, opportunity costs, and cognitive effort. This approach is also consistent with the theory of rational inattention, which posits that ``agents cannot process all available information, but they can choose which exact pieces of information to attend to'' \citep{matejka2023}.

The first part of the paper examines the incentives driving information acquisition. Decision-makers with different prior beliefs may exhibit different willingness to pay for a precise scrutiny, with some entirely unwilling to incur any cost for additional information. The shape of the willingness-to-pay function is determined by the specifics of the underlying decision problem and by the information received through a superficial scrutiny. 

Analyzing the belief patterns that emerge after optimal information acquisition is the focus of the paper's second part. The analysis begins by examining the relative belief patterns of two decision-makers, with a particular emphasis on belief polarization. The main proposition in this section establishes the necessary and sufficient conditions for polarization and derives the ex-ante probability of its occurrence. The analysis then shifts to the information acquisition choices and belief patterns of individual decision-makers. This part identifies the necessary and sufficient conditions for various phenomena, including disconfirmation, confirmatory belief patterns, and underreaction to information.

These results remark the importance of considering decision-makers' incentives for information acquisition when evaluating their responses to information. Beliefs can naturally polarize under common evidence because decision-makers with different priors make different information acquisition decisions. Ignoring these incentives may lead to choices and belief patterns that, from the perspective of an informed Bayesian observer, appear similar to well-documented biases in information processing, such as confirmation bias. These findings also offer several testable implications regarding beliefs and information acquisition choices, providing a foundation for empirical exploration.

This work is related to the literature on information acquisition and processing, and specifically to those papers which focus on patterns such as belief polarization and confirmation bias. \cite{calvert1985} uses a rational choice model to show how Bayesian decision-makers optimally select biased and imperfect advice based on their prior beliefs. In \cite{suen}, decision-makers acquire information through ``experts'' who coarsen continuous data into discrete recommendations. This process leads rational Bayesian individuals to gather information confirming their priors. Decision-makers with different prior beliefs may choose different experts, selecting those who coarsen information in ways that resonate with their priors. This behavior can foster the polarization of beliefs. 

\cite{nimark2019inattention} examine the persistence of disagreement over objective truths. They develop a model in which rational Bayesian agents learn about a state through noisy signals, and identify two key mechanisms driving disagreement: the confirmation effect, where agents select signal structures that reinforce their prior beliefs, and the complacency effect, where agents with more precise prior beliefs opt for less informative signals. When priors are sufficiently precise, agents may even choose completely uninformative signals. Together, these effects can lead to clustering of beliefs into opposing groups over time, despite agents being ex ante identical. Differently from the aforementioned work, this paper keeps the information structure fixed and focuses on decision-makers' optimal acquisition choices and belief patterns.

\cite{borgers} study the conditions under which two signals are substitutes or complements in the spirit of \cite{blackwell}. They show that the complementarity of signals is closely tied to the phenomenon where two agents, starting with differing prior beliefs, may become even more polarized after receiving additional public evidence. \cite{kondor} and \cite{andreoni} describe information structures in which increased polarization following the public disclosure of evidence is consistent with Bayesian updating. Polarization arises because agents disagree on how to interpret the public signal, implying that signals are complements according to \cite{borgers}.


The seminal paper by \cite{rabin} models confirmation bias as an exogenous probability that a decision-maker misinterprets signals contradicting her prior hypothesis. Specifically, the decision-maker unwittingly misinterprets only contrasting information, creating a distinction between the information she \emph{receives} and what she \emph{perceives}. Key implications of this bias in their model include overconfidence, belief polarization, and potential wrongness. The latter occurs when a Bayesian observer, aware of the agent’s confirmation bias, concludes that she supports an incorrect hypothesis. Similarly, this paper introduces a distinction between the information a decision-maker receives and the information she chooses to observe due to costly acquisition. A fully informed Bayesian observer, disregarding the cost constraints, may interpret the decision-maker’s behavior as irrational, despite its optimality within the model.


There is a tight connection between costly information acquisition and cognitive limitations \citep{dewatripont}. This connection resonates with psychological theories, such as the dual-process theory, which distinguish between two distinct modes of information processing. In this framework, individuals either rely on automatic, intuitive judgments or engage in more deliberate, reasoning-based processing \citep{kahneman}. The role of processing costs fits  with these theories, as cognitive limitations can constrain the ability to engage in deliberate yet straining processing. Furthermore, this idea connects to the Elaboration Likelihood Model (Petty and Cacioppo, 1986) and the Heuristic-Systematic Model (Chaiken, Liberman, and Eagly, 1989), which also propose that individuals process information through either heuristic shortcuts or more systematic, effortful strategies, depending on the cost and motivation to process information.\footnote{The Elaboration Likelihood Model and the Heuristic-Systematic Model both emphasize how individuals process persuasive messages either through a superficial, heuristic approach or through more effortful, systematic thinking. These models are foundational in understanding how information is processed under varying levels of motivation and cognitive resources.}


\subsection{An introductory example}\label{sec:example}

There are two possible states of the world, $A$ and $B$. Two decision-makers are uncertain about which state is true and must bet on it. Each decision-maker makes her independent guess and receives her payoff independently of the other decision-maker’s choice. If they guess correctly, they receive a payoff of $1$; otherwise, they receive $0$. Decision-maker \emph{High} ($H$) initially considers state $A$ more likely to be true, with prior belief $p^H(A) = 0.7$. In contrast, decision-maker \emph{Low} ($L$) initially believes state $B$ is more likely, with $p^L(A) = 0.3$. A bi-dimensional signal $\sigma = (\sigma_1, \sigma_2)$ is then publicly disclosed, consisting of two mutually independent components. Each dimension of $\sigma$ can take one of two possible values, $\alpha$ and $\beta$, where $\alpha$ supports $A$ as the true state of the world, and $\beta$ favors $B$.

Decision-makers can choose whether to scrutinize $\sigma$ 
superficially or precisely. A superficial acquisition is free of
charge but reveals only the realization of the first dimension, $\sigma_1$. By contrast, a full and precise acquisition
entails the payment of a processing cost $c>0$, which reveals the
realizations of both dimensions of $\sigma$. The processing cost can
be interpreted as additional time, effort, or simply a monetary amount
required to acquire $\sigma$ in its entirety. I make the
following assumptions regarding the informative content of the signals:
\begin{displaymath}
P(\sigma_1 = \alpha \mid  A) = P(\sigma_1 = \beta \mid  B)=\theta_1 = \nicefrac{3}{5},
\end{displaymath}
\begin{displaymath}
P(\sigma_2 = \alpha \mid  A) = P(\sigma_2 = \beta \mid  B)=\theta_2 = \nicefrac{4}{5}.
\end{displaymath}

Suppose that the processing cost is $c=\frac{1}{10}$ and the signal is
$\sigma=(\alpha , \beta)$. Both decision-makers acquire the first, free
component $\sigma_1=\alpha$, resulting in Bayesian posteriors of $p^H_{\alpha}(A)\approx0.78$
for $H$ and $p^L_{\alpha}(A)\approx0.39$ for $L$. Compared to their
prior beliefs, both agents have increased their support for state $A$. Furthermore, we observe \emph{convergence}, as
$p^H_{\alpha} - p^L_{\alpha} < p^H - p^L$. The decision-makers have different willingness to pay for observing $\sigma_2$. Specifically, $H$ is willing to pay at most $c^H\approx0.02$, while
$L$ is willing to pay at most $c^L\approx0.19$. Given that $c=\frac{1}{10}$, only $L$ opts for a precise acquisition of $\sigma$. After observing $\sigma_2=\beta$,
her posterior is $p^L_{\alpha \beta}\approx0.14$. 

The decision-makers' posterior beliefs now lie further apart compared to their priors, resulting in belief
\emph{divergence}. Moreover, while $H$ has increased her support for $A$,
$L$ has decreased it. This pattern is reminiscent of what is known as \emph{belief polarization}, as the decision-makers' beliefs diverge and move further apart even though they were endowed with the same signal. Despite the signal overall supports state $B$ because its second component is more informative than the first, decision-maker $H$ nonetheless increases her confidence in state $A$. Compared to the signal's informative content, she updates her beliefs in a way that confirms her previously held hypothesis about the state. This pattern is resembles what is known as \emph{confirmation bias}. 

Suppose now that the signal gives $\sigma_1=\beta$. In this case, the decision-makers would switch their willingness to pay for $\sigma_2$, with $c^H\approx0.19$ and $c^L\approx0.02$. Only decision-maker $H$ chooses to incur the processing cost to observe $\sigma_2$. This example highlights how decision-makers are more willing to scrutinize with greater accuracy information that contradicts their prior beliefs. This behavior recalls other biases in information selection and acquisition.\footnote{From \citet[p.~2099]{lord}: ``The biased assimilation processes underlying this
effect may include a propensity to [...] accept confirming evidence
at face value while scrutinizing disconfirming
evidence hypercritically.''} Finally, suppose instead that $\sigma=(\alpha,\alpha)$. We have seen before that, after observing $\sigma_1=\alpha$, decision-maker $H$ decides not to incur the processing cost to observe $\sigma_2$, and thus her posterior is $p^H_{\alpha}(A)\approx0.78$. However, if she instead observes $\sigma_2=\alpha$, her posterior would be $p^H_{\alpha\alpha}(A)\approx0.93$. This belief pattern, compared to the signal's whole informational content, is reminiscent of \emph{under-reaction} to information, as $0.78<0.93$.

\section{The model}\label{sec:model}

\subsection{Set-up}\label{sec:setup}

There is a state of the world represented as a random variable, $\tilde{\omega}$, with realization $\omega \in\Omega = \{ A,B \}$. A decision-maker (DM, she) must take an action $s\in\Omega$, which is interpreted as a literal guess of the realized state. The decision-maker believes that the state $\omega=A$ realizes with prior probability $p=Pr\left(\omega=A\right)\in[0,1]$. 

Before choosing an action $s\in\Omega$, the decision-maker is endowed with a two-dimensional signal, $\sigma$. The signal is a random variable $\tilde\sigma=\left(\tilde\sigma_1,\tilde\sigma_2 \right)$ with realization $\sigma=(\sigma_1 , \sigma_2)\in \Theta^2$, where $\Theta= \{ \alpha, \beta\}$. I indicate the conditional and unconditional probability distribution of $\sigma$ by $P$. Each component $\sigma_j$, $j\in\{1,2\}$, is itself an informative signal about the state, and $\sigma_1$ and $\sigma_2$ are mutually independent. The information structure is
\begin{displaymath}
\theta_1:= P\left(\tilde\sigma_1 = \alpha \mid  \tilde{\omega} =A \right) = P\left(\tilde\sigma_1 = \beta \mid  \tilde{\omega}= B\right) >\nicefrac{1}{2},
\end{displaymath}
\begin{displaymath}
\theta_2 := P\left(\tilde\sigma_2 = \alpha \mid  \tilde{\omega} = A \right) = P\left(\tilde\sigma_2 = \beta \mid  \tilde{\omega} =B \right) > \nicefrac{1}{2}.
\end{displaymath}
I will sometimes use $\alpha_j$ and $\beta_j$ to denote $\sigma_j=\alpha$ and $\sigma_j=\beta$, respectively.

The signal's first component, $\sigma_1$, is freely observable. By contrast, its second component, $\sigma_2$, can be scrutinized only by incurring a processing cost $c>0$. After receiving $\sigma$, but before choosing $s\in\Omega$, the decision-maker chooses whether to observe only the signal's first component for free, or to acquire also its second component at a cost. Formally, this choice is represented by action $a \in \Pi=\{
\rho, \neg\rho\}$, where $\rho$ indicates the decision of paying $c$ to
observe both $\sigma_1$ and $\sigma_2$, and $\neg\rho$ the decision of acquiring only $\sigma_1$. Action $a$ is selected after observing $\sigma_1$ but before choosing $s$, thus $a(p,c,\sigma_1):[0,1]\times \mathbb{R}^+ \times\Theta\to\Pi$. Conversely, $s(p,\sigma_1):[0,1]\times \Theta\to\Omega$ if $a=\neg\rho$, and $s(p,\sigma_1,\sigma_2):[0,1]\times \Theta^2\to\Omega$ otherwise.

The decision-maker obtains a utility of $v(s=\omega,\omega)=\overline{U}$ when she selects an action $s$ that matches with the state's realization $\omega$, and obtains $v(s\neq\omega,\omega)=\underline{U}$ otherwise. I denote $\Delta U :=
\overline{U} -\underline{U}$, and assume that $\Delta U > 0$. The decision-maker's utility when the state is $\omega\in\Omega$ and her actions are $a\in \Pi$ and $s\in\Omega$ is $u(a, s,\omega)$, where\footnote{Formally, $a(p,c,\sigma_1)$}
\[
 u(a, s,\omega) =
  \begin{cases}
   v(s,\omega)    & \text{if } a=\neg \rho, \\
   v(s,\omega) - c  & \text{otherwise}.
  \end{cases}
\]

The decision-maker updates information according to Bayes' rule. Her posterior beliefs after receiving $\sigma$ is denoted by $p_{\sigma_1}$ when $a=\neg \rho$, and by either $p_{\sigma_1\sigma_2}$ or $p_{\sigma}$ otherwise. The decision-maker selects actions $a\in \Pi$ and $s\in \Omega$ to maximize her expected utility. I indicate her optimal decisions given beliefs by $a^*\in \Pi$ and $s^*\in\Omega$. When choosing $a\in\Pi$, her beliefs are necessarily $p_{\sigma_1}$. Denote by $p(\sigma)$ the decision-maker's posterior beliefs at the time of choosing $s\in\Omega$, and conditional on optimal information acquisition choice, $a^*$. We have that $p(\sigma)=p_{\sigma_1}$ if $a^*=\neg\rho$, and $p(\sigma)=p_{\sigma}$ otherwise. Therefore,  
\[
a^* := \argmax_{a \in \Pi} \mathbb{E}_{p_{\sigma_1}}\left[u(a, s^*,\omega)\mid  \sigma\right], \; \text{ and } \; s^* :=
\argmax_{s \in \Omega} \mathbb{E}_{p(\sigma)}\left[u(a^*, s,\omega)\mid \sigma\right].
\]


Part of the analysis will consider two decision-makers, say $i$ and $j$, that differ in their prior beliefs only. In those cases, I indicate their prior beliefs with $p^i$ and $p^j$, and their posterior beliefs with $p^k_{\sigma_1}$, $p^k_{\sigma_1 \sigma_2}$, and $p^k(\sigma)$ for $k\in\{i,j\}$. Similarly, I will use $a_k$, $a^*_k$, $s_k$ and $s^*_k$ for their actions, and will always imply that $i\neq j$. From the perspective of decision-maker~$k$, with prior $p^k$, the conditional distribution of the signal's second component is $P^k(\sigma_2 \mid\sigma_1)$, and the unconditional distributions are $P^k(\sigma_j)$.

The model's timeline, as depicted by Figure~\ref{fig:timing}, is as follows: (i) The state, $\omega\in\Omega$, and the signal, $\sigma\in\Theta^2$, are realized but not observed; (ii) The decision-maker observes $\sigma_1$ for free, and then chooses $a\in \Pi$; (iii) If $a=\rho$, then the decision-maker observes $\sigma_2$ and incurs a cost $c$; (iv) The decision-maker selects $s\in\Omega$; (v) Finally, her payoff realizes.

\begin{figure}
    \centering
    \includegraphics[width=0.95\linewidth]{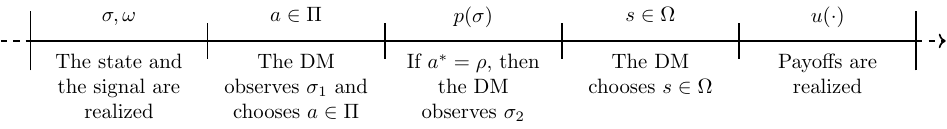}
    \caption{Timing structure}
    \label{fig:timing}
\end{figure}

\subsection{Discussion of the model}

In this model, the decision-maker faces a problem of costly information acquisition. The available signal consists of two components: the first is free, while the second is costly to access. The analysis initially focuses on the decision-makers' incentives for acquiring information. Then, it shifts to studying the relative beliefs of decision-makers who have different prior beliefs but are otherwise identical. Decision-makers always scrutinize the first signal because it is free. Whether they choose to acquire the second, costly component depends crucially---but not solely---on their prior beliefs.

Some findings could be derived from a simpler model, where there is a single, one-dimensional signal that is costly to acquire. However, such a simplified framework would not capture the variety of results and belief patterns that the current two-dimensional model provides. A model with more than two dimensions would add additional complexities without providing more insights, as a rich variety of belief patterns is addressed by the current two-dimensional setting. Moreover, the signals are assumed to be symmetric. While asymmetric signals may offer additional benefits to decision-makers \citep{calvert1985}, removing symmetry would complicate the analysis without yielding significant insights.

In the analysis, decision-makers differ only in their prior beliefs, while remaining identical in all other respects. This assumption serves three purposes. First, it mirrors the typical setup of controlled laboratory experiments, where subjects face the same payment scheme. Second, it isolates the role of prior beliefs in shaping incentives for information acquisition and the resulting belief patterns. Finally, it simplifies the analysis without significant loss of generality. Decision-makers with different utility functions could, of course, make different decisions even with the same prior beliefs, but this possibility is excluded in the current setting.


\section{Incentives for information processing}\label{sec:incentives}

This section examines the incentives for information acquisition faced by a decision-maker. The first part (Section~\ref{sec:differentcases}) categorizes a decision-maker's problem into eight qualitatively distinct cases based on her prior beliefs. This categorization is instrumental for the subsequent analysis and establishes a link between the decision-maker's prior beliefs and their willingness to pay for additional information. The section illustrates how to calculate the willingness to pay for decision-makers belonging to the first case, while the remainder of the analysis is presented in the appendix. The second part (Section~\ref{sec:thresholdcostfunction}) investigates the function that describes such a willingness to pay, conditional on the decision-maker's prior beliefs, signal received, and information structure. This analysis clarifies the relationship between prior beliefs and the value of additional information.

\subsection{Prior categorization}\label{sec:differentcases}

A decision-maker's incentive to scrutinize $\tilde\sigma_2$ depends on her interim beliefs ($p_{\sigma_1}$), the payoff structure ($\Delta U$), the processing cost ($c$), and the information structure of the signal's second component ($\theta_2$). Given the informative content of each dimension of \(\tilde\sigma\), every prior $p$ induces posteriors with implications on the DM's optimal behavior that are crucial in determining her willingness to pay for additional information.

Consider, for example, a decision-maker with prior belief $p$, such that, after observing $\sigma_1 = \alpha$, her posterior belief is greater than $\nicefrac{1}{2}$ independently of whether she observes $\tilde\sigma_2$ and its possible realizations. As a result, she will choose action $A$ regardless of $\sigma_2$. I will refer to this situation as ``case~1,'' where $p_{\alpha}$, $p_{\alpha \alpha}$, and $p_{\alpha \beta}$ are all greater than or equal to $\nicefrac{1}{2}$. Although the costly signal $\tilde\sigma_2$ is informative about the true state of the world $\tilde\omega$, the DM's choice is independent of its realization. In this case, the DM cannot benefit from acquiring the signal's second component, and her willingness to pay for observing $\sigma_2$ is zero. 

Intuitively, decision-makers whose prior falls into \emph{case~1} are relatively extreme in the sense that they place a very high (prior) probability on the state being $\omega=A$. Decision-makers with more moderate prior beliefs may have different incentives for information acquisition. To aid the analysis that follows, it is useful to partition the decision-maker's belief space in a way that links each partition to a qualitatively different optimal behavior. Recall that the decision-maker optimally selects action $s=A$ if her posterior beliefs exceed $\nicefrac{1}{2}$, and selects action $s=B$ otherwise. Table~\ref{tab:cases:posterior} groups all posterior beliefs that lead to qualitatively different choices over $s\in\Omega$ conditional on all possible observed signals' realizations. Each group of posteriors induces a convex set of prior beliefs, and constitutes a different \emph{case}. This procedure partitions the belief space into eight cases. 

Figure~\ref{fig:partitions} illustrates,\ for a given information structure, an example of partition  induced by the categorization as in Table~\ref{tab:cases:posterior}. The next section illustrates the connection between cases and decision-makers' willingness to pay for observing the signal's second component.

\begin{figure}
    \centering
    \includegraphics[width=0.75\linewidth]{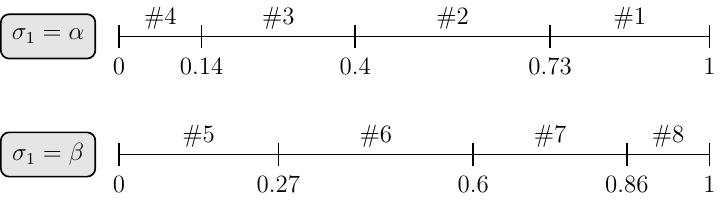}
    \caption{Partition of the conditional posterior belief spaces induced by an information structure where $\theta_1 = \nicefrac{3}{5}$ and $\theta_2 = \nicefrac{4}{5}$, categorized according to Table~\ref{tab:cases:posterior}. The case numbers are labeled at the top of each corresponding partition.}
    \label{fig:partitions}
\end{figure}



\begin{table}[htbp]
\begin{center}
\begin{tabular}{llll}
\toprule
\multicolumn{4}{c}{\textbf{Posterior beliefs and optimal behavior}}\\

\multicolumn{1}{c}{Case $\#$}& \multicolumn{1}{c}{$\tilde\sigma_1$} &
\multicolumn{1}{c}{Posteriors'properties} & \multicolumn{1}{c}{Prior's
interval} \\
\midrule
\multicolumn{1}{c}{1} & $\alpha_1$  &$p_{\alpha} \geq \frac{1}{2}$, $p_{\alpha \alpha} \geq
  \frac{1}{2}$, $p_{\alpha \beta} \geq \frac{1}{2}$ & $ \left[\frac{\theta_2 (1- \theta_1)}{\theta_1 + \theta_2 - 2 \theta_1
  \theta_2},1\right]$ \\

\multicolumn{1}{c}{2} & $\alpha_1$  &$p_{\alpha} \geq \frac{1}{2}$, $p_{\alpha \alpha} \geq
  \frac{1}{2}$, $p_{\alpha \beta} \leq \frac{1}{2}$ & $ \left[  1-\theta_1 , \frac{\theta_2 (1- \theta_1)}{\theta_1 + \theta_2 - 2 \theta_1
  \theta_2} \right]$  \\

\multicolumn{1}{c}{3} & $\alpha_1$  &$p_{\alpha} \leq \frac{1}{2}$, $p_{\alpha \alpha} \geq
  \frac{1}{2}$, $p_{\alpha \beta} \leq \frac{1}{2}$ & $ \left[  1 - \frac{\theta_1 \theta_2 }{1- \theta_1 - \theta_2 + 2 \theta_1
  \theta_2} , 1- \theta_1 \right]$ \\

\multicolumn{1}{c}{4} & $\alpha_1$  &$p_{\alpha} \leq \frac{1}{2}$, $p_{\alpha \alpha} \leq
  \frac{1}{2}$, $p_{\alpha \beta} \leq \frac{1}{2}$&  $ \left[ 0, 1 - \frac{\theta_1 \theta_2 }{1- \theta_1 - \theta_2 + 2 \theta_1
  \theta_2} \right]$ \\

\multicolumn{1}{c}{5} & $\beta_1$  &$p_{\beta} \leq \frac{1}{2}$, $p_{\beta \alpha} \leq
  \frac{1}{2}$, $p_{\beta \beta} \leq \frac{1}{2}$&  $ \left[0, \frac{\theta_1 (1- \theta_2 )}{ \theta_1 + \theta_2 - 2 \theta_1
  \theta_2}\right]$ \\

\multicolumn{1}{c}{6} & $\beta_1$  &$p_{\beta} \leq \frac{1}{2}$, $p_{\beta \alpha} \geq
  \frac{1}{2}$, $p_{\beta \beta} \leq \frac{1}{2}$&  $ \left[  \frac{\theta_1 (1- \theta_2 )}{ \theta_1 + \theta_2 - 2 \theta_1
  \theta_2} , \theta_1 \right]$ \\

\multicolumn{1}{c}{7}& $\beta_1$ & $p_{\beta} \geq \frac{1}{2}$, $p_{\beta \alpha} \geq
  \frac{1}{2}$, $p_{\beta \beta} \leq \frac{1}{2}$ & $ \left[ \theta_1 , \frac{\theta_1 \theta_2 }{ 1 - \theta_1 - \theta_2 + 2 \theta_1
  \theta_2} \right]$ \\

\multicolumn{1}{c}{8}& $\beta_1$  & $p_{\beta} \geq \frac{1}{2}$, $p_{\beta \alpha} \geq
  \frac{1}{2}$, $p_{\beta \beta} \geq \frac{1}{2}$&  $ \left[ \frac{\theta_1 \theta_2 }{ 1 - \theta_1 - \theta_2 + 2 \theta_1
  \theta_2} , 1 \right] $ \\

\bottomrule
\end{tabular}
\caption{Categorization of prior beliefs in eight cases yielding different optimal behavior conditional on observed information.}
\label{tab:cases:posterior}
\end{center}
\end{table}



\begin{table}[htbp]
\begin{center}
\begin{tabular}{llll}
\toprule
\multicolumn{4}{c}{\textbf{Willingness to pay for additional information}}\\

\multicolumn{1}{c}{Case $\#$}& \multicolumn{1}{c}{$\tilde\sigma_1$}
& \multicolumn{1}{c}{Prior's interval} &\multicolumn{1}{c}{WTP $c_{\sigma_1}(p)$} \\
\midrule
\multicolumn{1}{c}{1} & $\alpha_1$  & $ \left[\frac{\theta_2 (1- \theta_1)}{\theta_1 + \theta_2 - 2 \theta_1
  \theta_2},1 \right]$ & \multicolumn{1}{c}{$0$} \\

\multicolumn{1}{c}{2} & $\alpha_1$  & $ \left[  1-\theta_1 , \frac{\theta_2 (1- \theta_1)}{\theta_1 + \theta_2 - 2 \theta_1
  \theta_2} \right]$&$\Delta U \cdot \left[ \frac{\theta_2 (1-p) - \theta_1 p -
    \theta_1 \theta_2 (1-2p)}{\theta_1 p + (1-\theta_1)(1-p)} \right]$
  \\

\multicolumn{1}{c}{3} & $\alpha_1$  &$ \left[  1 - \frac{\theta_1 \theta_2 }{1- \theta_1 - \theta_2 + 2 \theta_1
  \theta_2} , 1- \theta_1 \right]$ &$\Delta U \cdot \left[
  \frac{\theta_1 \theta_2 p - (1-\theta_1)(1-\theta_2)(1-p)}{\theta_1
    p + (1-\theta_1)(1-p)} \right] $  \\

\multicolumn{1}{c}{4} & $\alpha_1$  & $ \left[ 0, 1 - \frac{\theta_1 \theta_2 }{1- \theta_1 - \theta_2 + 2 \theta_1
  \theta_2} \right]$ &
  \multicolumn{1}{c}{$0$}   \\

\multicolumn{1}{c}{5} & $\beta_1$ &  $ \left[0, \frac{\theta_1 (1- \theta_2 )}{ \theta_1 + \theta_2 - 2 \theta_1
  \theta_2}\right]$ &\multicolumn{1}{c}{$0$} \\

\multicolumn{1}{c}{6} & $\beta_1$ &  $ \left[  \frac{\theta_1 (1- \theta_2 )}{ \theta_1 + \theta_2 - 2 \theta_1
  \theta_2} , \theta_1 \right]$ & $\Delta U \cdot \left[ \frac{(1-\theta_1)\theta_2p -
   \theta_1(1-\theta_2)(1-p)}{(1-\theta_1)p + \theta_1(1-p)} \right]$ \\

\multicolumn{1}{c}{7}& $\beta_1$ & $ \left[ \theta_1 , \frac{\theta_1 \theta_2 }{ 1 - \theta_1 - \theta_2 + 2 \theta_1
  \theta_2} \right]$ & $\Delta U \cdot \left[ \frac{\theta_1 \theta_2(1-p) -
   (1- \theta_1)(1-\theta_2)p}{(1-\theta_1)p + \theta_1(1-p)} \right]$ \\

\multicolumn{1}{c}{8}& $\beta_1$  &  $ \left[ \frac{\theta_1 \theta_2 }{ 1 - \theta_1 - \theta_2 + 2 \theta_1
  \theta_2} , 1 \right] $&\multicolumn{1}{c}{$0$}  \\

\bottomrule
\end{tabular}
\caption{The DM's willingness to pay for observing $\sigma_2$ given prior beliefs $p$ and signal realization $\sigma_1$.}
\label{tab:cases:cost}
\end{center}
\end{table}



\subsection{Willingness to pay for additional information}\label{sec:thresholdcostfunction}

The previous section categorized decision-makers by cases. Each case is characterized by the the pair $(p,\sigma_1)$, which induces the interim posterior $p_{\sigma_1}$. This section describes and analyzes the \emph{cost function} characterizing the decision-makers' willingness to pay for observing the signal's second component, $\sigma_2$.  As we shall see, the shape and the properties of this function vary across cases. The proofs of Propositions~\ref{prop:cost_function} and \ref{prop:cost_properties}, which are outlined in this section, are in Appendix~\ref{sec:cases} and \ref{sec:costfunction}, respectively.

The decision-maker's interim posterior beliefs, $p_{\sigma_1}$, determine her willingness to pay for observing $\sigma_2$. 
Given a fixed payoff and information structure, I hereafter refer to such a willingness to pay by $c_{\sigma_1}(p)$, and say that a processing cost is \emph{admissible} when it is smaller than $c_{\sigma_1}(p)$. Table~\ref{tab:cases:cost} associates each case to the decision-maker's willingness to pay for observing $\sigma_2$, that is, $c_{\sigma_1}(p)$.

\begin{definition}\label{def:costfunction}
    The function $c_{\sigma_1}(p):[0,1]\times \Theta_1 \mapsto \mathbb{R}^+_0$ represents the highest processing cost that a decision-maker with prior beliefs $p$ is willing to incur after observing $\sigma_1$. A processing cost, $c$, is ``admissible'' if $c\leq c_{\sigma_1}(p)$. The set $\mathcal{C}_{\sigma_1}(p):=[0, c_{\sigma_1}(p)]$ contains all admissible processing costs given prior beliefs $p$ and signal realization $\sigma_1$.
\end{definition}


 The function $c_{\sigma_1}(p)$, as in Definition~\ref{def:costfunction}, is piecewise with discontinuities between different cases, and is characterized by the following proposition. I denote by $c^n_{\sigma_1}(p)$ the cost function for a decision-maker that
belongs to case $n$.\footnote{The notation for the cost function $c^n_{\sigma_1}(p)$ is redundant, as each case number is already associated with $\sigma_1$. However, $\sigma_1$ is retained for clarity and conformity with Definition~\ref{def:costfunction}.}

\begin{proposition}\label{prop:cost_function}
    Given an information structure $(\theta_1,\theta_2)$, the decision-maker's willingness to pay for observing $\sigma_2$ is determined by the cost function $c_{\sigma_1}(p)$, where 
    \[ c_{\alpha}(p) = \left\{
 \begin{array}{l l}

 c^4_{\alpha}(p)= 0 & \quad \text{if $p \in  \left[ 0, 1 - \frac{\theta_1 \theta_2 }{1- \theta_1 - \theta_2 + 2 \theta_1
  \theta_2} \right]$},\\

 c^3_{\alpha}(p)=  \Delta U \cdot \left[
  \frac{\theta_1 \theta_2 p - (1-\theta_1)(1-\theta_2)(1-p)}{\theta_1 
    p + (1-\theta_1)(1-p)} \right] & \quad \text{if $p \in \left[  1 - \frac{\theta_1 \theta_2 }{1-
      \theta_1 - \theta_2 + 2 \theta_1 
  \theta_2} , 1-\theta_1 \right]$}, \\ 

  c^2_{\alpha}(p)=  \Delta U \cdot \left[ \frac{\theta_2 (1-p) - \theta_1 p -  
    \theta_1 \theta_2 (1-2p)}{\theta_1 p + (1-\theta_1)(1-p)} \right]
& \quad \text{if $p \in \left[  1-\theta_1 , \frac{\theta_2 (1- \theta_1)}{\theta_1 + \theta_2 - 2 \theta_1
  \theta_2} \right]$},\\

c^1_{\alpha}(p)= 0 & \quad \text{if $p \in \left[ \frac{ \theta_2 (1-
           \theta_1)}{\theta_1 + \theta_2 - 2 \theta_1 
  \theta_2} ,1 \right]$},
\end{array}
\right.\]

\[ c_{\beta}(p) = \left\{
 \begin{array}{l l}

 c^5_{\beta}(p) = 0 & \quad \text{if $p \in  \left[ 0, \frac{\theta_1 (1- \theta_2 )}{ \theta_1 + \theta_2 - 2 \theta_1
  \theta_2} \right] $},\\

  c^6_{\beta}(p) = \Delta U \cdot \left[ \frac{(1-\theta_1)\theta_2p -
   \theta_1(1-\theta_2)(1-p)}{(1-\theta_1)p + \theta_1(1-p)} \right] &
\quad \text{if $p \in \left[ \frac{\theta_1 (1- \theta_2 )}{ \theta_1 + \theta_2 - 2 \theta_1
  \theta_2} , \theta_1 \right]$}, \\ 

   c^7_{\beta}(p) = \Delta U \cdot \left[ \frac{\theta_1 \theta_2(1-p) -
   (1- \theta_1)(1-\theta_2)p}{(1-\theta_1)p + \theta_1(1-p)} \right]
& \quad \text{if $p \in \left[ \theta_1 , \frac{\theta_1 \theta_2 }{ 1 - \theta_1 - \theta_2 + 2 \theta_1
  \theta_2} \right]$},\\

c^8_{\beta}(p) = 0 & \quad \text{if $p \in \left[ \frac{\theta_1 \theta_2 }{ 1 - \theta_1 - \theta_2 + 2 \theta_1
  \theta_2} ,1 \right]$}.
\end{array}
\right.\]
\end{proposition}

\begin{figure}[htbp]
        \centering
     \includegraphics[width=\textwidth]{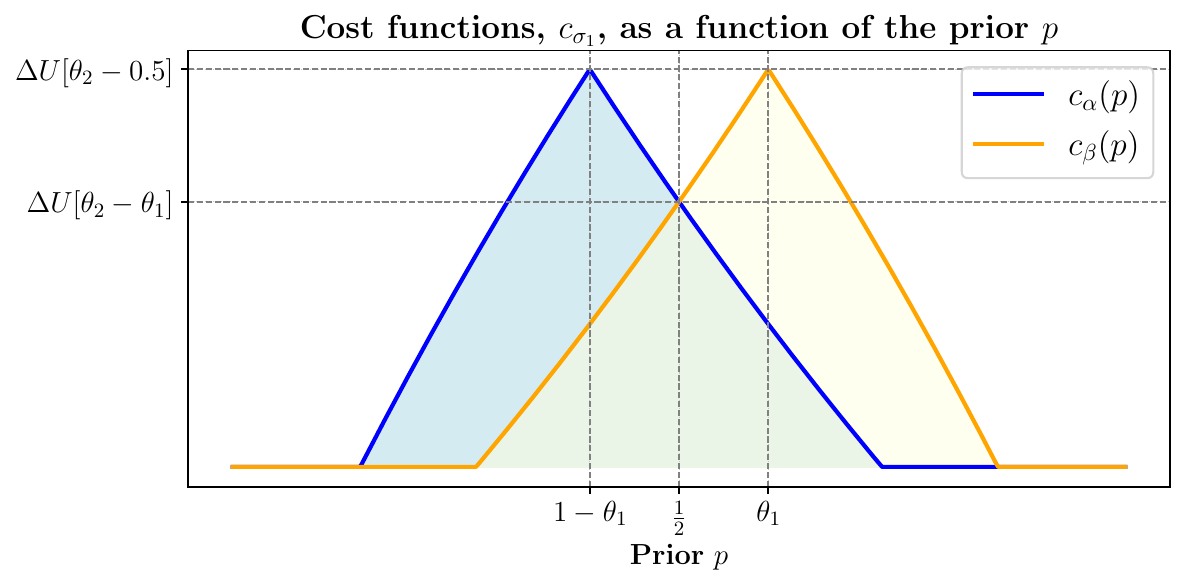}   
        \caption{The cost function for different prior beliefs: the blue line represents $c_{\alpha}(p)$, and the yellow line represents $c_{\beta}(p)$. The shaded areas indicate the respective sets of admissible costs, i.e., $\mathcal{C}_{\alpha}(p)$ and $\mathcal{C}_{\beta}(p)$.}\label{fig:costfunction}
\end{figure}

\begin{proposition}\label{prop:cost_properties}
    The cost function $c_{\sigma_1}(p)$,
    \begin{itemize}[noitemsep]
        \item is continuous in $p$, and satisfies $c_{\alpha}(p)=c_{\beta}(1-p)$ for all $p\in[0,1]$;
        \item $c^3_{\alpha}(p)$ is strictly increasing and concave, while $c^2_{\alpha}(p)$ is strictly decreasing and convex;
        \item $c^6_{\beta}(p)$ is strictly increasing and convex, while $c^7_{\beta}(p)$ is strictly decreasing and concave;
        \item $c_{\alpha}(p)$ has a global maximum at $p=1-\theta_1$, while $c_{\beta}(p)$ has a global maximum at $p=\theta_1$, where $c_{\alpha}(1-\theta_1)=c_{\beta}(\theta_1)=\Delta U \left[\theta_2 - \frac{1}{2} \right]$.
    \end{itemize}
\end{proposition}

Figure~\ref{fig:costfunction} illustrates the decision-maker's willingness to pay for observing the signal's second component, $\tilde\sigma_2$, conditional on their prior beliefs and on the realization of the signal's first component, $\sigma_1$. 

\subsection{Prior beliefs and information processing}

Which prior beliefs would induce decision-makers to incur the cost for processing additional information? The cost function characterization from the previous section allows us to answer this question. First, Proposition~\ref{prop:cost_properties}  shows that no decision-maker is willing to acquire $\sigma_2$ when the processing cost exceeds $\Delta U \left[\theta_2 - \frac{1}{2} \right]$. For more reasonable costs, $0<c<\Delta U \left[\theta_2 - \frac{1}{2} \right]$, Propositions~\ref{prop:cost_function} and \ref{prop:cost_properties}, along with a visual inspection of Figure~\ref{fig:costfunction}, suggest that only decision-makers with prior beliefs in convex and strict subsets of the belief space would pay the processing cost to observe $\tilde\sigma_2$. Moreover, these sets must depend on $\sigma_1$, as the decision to acquire $\tilde\sigma_2$ hinges on the interim posterior, $p_{\sigma_1}$.

\begin{definition}\label{def:hulsets}
   The set of prior beliefs for which a decision-maker, after observing $\sigma_1$, is willing to incur a given processing cost to observe $\sigma_2$ is defined as
\begin{equation*}
    \mathcal{H}_{\sigma_1}(c) \coloneqq \left\{ p \in [0,1] \mid  c_{\sigma_1}(p) > c \right\}.
\end{equation*} 
\end{definition}

The following result confirms the intuition that only decision-makers with relatively moderate prior beliefs are willing to acquire additional information at a cost, while those with extreme beliefs---closer to zero or one---do not find it beneficial. However, this intuition holds only if the signal's second component is more informative than the first. Otherwise, even decision-makers with prior beliefs near $\frac{1}{2}$ may not benefit from observing the signal's second component. Furthermore, some ``extreme'' decision-makers may remain unwilling to scrutinize additional information, even as the processing cost approaches zero.

\begin{corollary}\label{cor:processing_priors}
Given a processing cost $c$, we have that 
\[ \mathcal{H}_{\sigma_1}(c)= \left\{
 \begin{array}{l l}

 \varnothing & \quad \text{if $c >\Delta U \left[ \theta_2 - \frac{1}{2}  \right]$},\\

 \left(\underline{\smash q}_{\sigma_1}(c), \overline{q}_{\sigma_1}(c)  \right) & \quad \text{if $c \in \left( 0 , \Delta U \left[ \theta_2 - \frac{1}{2}  \right] \right)$}, 
\end{array}
\right.\]
where $0 < \underline{\smash q}_{\alpha} (c) <\underline{\smash q}_{\beta} (c)$ and $\overline{q}_{\alpha}(c) < \overline{q}_{\beta}(c) < 1$. Moreover, $\underline{\smash q}_{\beta} (c)\leq \overline{q}_{\alpha}(c)$ if and only if $\theta_2\geq\theta_1$.
\end{corollary}

The corollary above additionally shows that the sets $\mathcal{H}_{\sigma_1}(c)$ are convex, with $\mathcal{H}_{\alpha}(c)$ shifted relatively to the left with respect to $\mathcal{H}_{\beta}(c)$. The boundaries of $\mathcal{H}_{\sigma_1}(c)$ are derived by inverting the decision-makers' cost function. Explicit closed-form formulas for $\underline{\smash q}_{\sigma_1}(c)$ and $ \overline{q}_{\sigma_1}(c) $ are provided in Appendix \ref{sec:inversecostfunctions}.


\section{Analysis of belief patterns}\label{sec:rationality}

This section examines the relative belief patterns exhibited by either one or two decision-makers. Section~\ref{sec:intermediate} introduces additional definitions and partial results to support the analysis that follows. Section~\ref{sec:polarization} focuses the relative belief patterns of two decision-makers, $i$ and $j$, with prior beliefs $p^i$ and $p^j$.  In that section, the only potential difference between the decision-makers lies in their beliefs, while both share the same signal realization. By contrast, Section~\ref{sec:behaviors} analyzes the belief patterns of a single decision-maker.

\subsection{Additional definitions and intermediate results}\label{sec:intermediate}

This section introduces definitions and intermediate results that are essential for the upcoming analysis. Four main concepts are defined: \emph{diverging attitudes}, \emph{inverse updating}, \emph{extreme}, and \emph{reciprocal} beliefs. Two decision-makers exhibit diverging attitudes (DA) when the distance between their posteriors exceeds the distance between their priors. They exhibit inverse updating (IU) when one decision-maker increases her belief in one state of the world, while the other decreases hers.\footnote{In~\cite{chang}, this concept is referred to as ``contrary updating.'' As they point out, inverse updating can lead to both divergence and convergence of beliefs, but not to ``parallel updating.''}

More formally, consider the Euclidean distance between the decision-makers' prior beliefs and that between their posterior beliefs, that is, respectively, $| p^i-p^j |$, and $| p^i(\sigma)-p^j(\sigma) |$. The \emph{divergence function}, defined as 
\begin{equation}\label{eq:div}
D(\sigma) \coloneqq | p^i-p^j | - | p^i(\sigma)-p^j(\sigma) |,
\end{equation}
indicates whether the decision-makers' posterior beliefs are closer or farther from each other compared to their priors when the signal realization is $\sigma$. A positive (resp. negative) divergence function indi that their beliefs became closer (resp. farther) after receiving $\sigma$. The \emph{inversion function}, defined as
\begin{equation}\label{eq:inv}
I(\sigma) \coloneqq [p^i - p^i(\sigma)] \cdot [p^j - p^j(\sigma)],
\end{equation}
indicates whether the decision-makers update their posterior beliefs in the same direction when the signal is $\sigma$. A positive (resp. negative) inversion function means that decision-makers update their prior in the same (resp. opposite) direction. 

\begin{definition}\label{def:daiu}
Diverging attitudes occur when $D(\sigma)<0$; inverse updating when $I(\sigma)<0$.
\end{definition}

A necessary condition for diverging attitudes and inverse updating is that one decision-maker incurs the processing cost to scrutinize $\sigma$ as a whole, whereas the other does not and observes $\sigma_1$ only. Formally, $a_i^*(\sigma_1)\neq a_j^*(\sigma_1)$. Otherwise, their beliefs would converge and move in the same direction. This implies that it is necessary for the occurrence of DA and IU that decision-makers have different prior beliefs, as otherwise they would have exactly the same incentives for information processing (see Section~\ref{sec:incentives}). As we shall see, the necessary conditions for DA and IU are stronger than just $p^i\neq p^j$. 

In Section~\ref{sec:incentives}, we observed that decision-makers with prior beliefs sufficiently close to either zero or one do not find it profitable to incur the processing cost which is required to observe the signal's second component. Intuitively, these decision-makers are so confident about the realized state that they are unwilling to pay any cost for new information, regardless of how small the cost is. In this sense, they hold \emph{extreme beliefs}. The next definition formalizes the sets containing extreme and non-extreme beliefs. 

\begin{definition}\label{def:extremeset}
    The set of non-extreme beliefs is
    \[
    \neg \mathcal{E} \coloneqq \lim_{c\to 0} \mathcal{H}_\alpha (c) \cup \mathcal{H}_\beta (c),
    \]
    while the set of extreme beliefs is $\mathcal{E}\coloneqq[0,1]\setminus\neg\mathcal{E}$. The sets of non-extreme and extreme beliefs conditional on observing $\sigma_1$ are, respectively, 
    \[
    \neg\mathcal{E}_{\sigma_1}\coloneqq \lim_{c\to 0} \mathcal{H}_{\sigma_1} (c) \; \text{ and } \; \mathcal{E}_{\sigma_1}\coloneqq[0,1]\setminus\neg\mathcal{E}_{\sigma_1}.
    \]
\end{definition}

The sets of extreme and non-extreme beliefs are intimately related with the categorization outlined in Table \ref{tab:cases:posterior}. Using the results from Section \ref{sec:incentives}, we can find a closed-form solution for these sets. Specifically, we have $\neg\mathcal{E}=\neg\mathcal{E}_{\alpha}\cup\neg\mathcal{E}_{\beta}$, where
\begin{displaymath}
 \neg\mathcal{E}_{\alpha} \triangleq \left(\underline{\smash q}_{\alpha}(0), \overline{q}_{\alpha}(0)  \right) = \left(  \frac{(\theta_1-1) (\theta_2-1)}{1 - \theta_1 - \theta_2 +     2 \theta_1 \theta_2} , \frac{(1- \theta_1)\theta_2}{\theta_1 +    \theta_2 - 2\theta_1 \theta_2}   \right),
\end{displaymath}
\begin{displaymath}
 \neg\mathcal{E}_{\beta}  \triangleq \left(\underline{\smash q}_{\beta}(0), \overline{q}_{\beta}(0)  \right)   = \left(  \frac{\theta_1 (1-\theta_2)}{ \theta_1 + \theta_2 - 2    \theta_1 \theta_2} , \frac{\theta_1 \theta_2}{1 - \theta_1 - \theta_2 +    2 \theta_1 \theta_2}  \right).
\end{displaymath}

The set $\neg\mathcal{E}$ is convex if and only if the signal's second component is more informative than the first one (i.e., $\theta_2\geq \theta_1$). Otherwise, $\neg\mathcal{E}$ would have a gap in the middle where beliefs are considered ``extreme'' despite being relatively moderate. Decision-makers with extreme beliefs belong to either cases~4 and 5, or cases~1 and 8 when $\theta_2 \geq \theta_1$. They can belong to both cases~1 and 5 when $\theta_2 < \theta_1$. In the latter scenario, the second signal is not sufficiently informative to alter the best course of action for decision-makers with relatively central prior beliefs after observing the first signal's realization.

From Section~\ref{sec:incentives}, we have the straightforward result that two decision-makers sharing the same prior beliefs also share identical incentives for information acquisition. Either both opt to pay for observing the signal's second component, or both abstain from doing so. They face the same set of admissible costs. However, sharing the same prior beliefs is not a necessary condition for aligned incentives in information acquisition. There are decision-makers who, despite having different prior beliefs, end up with the same set of admissible costs after observing $\tilde\sigma_1$. This is trivially true for decision-makers with extreme beliefs who do not benefit from additional information, but it is also true for pairs of decision-makers who do not hold extreme beliefs. I refer to such decision-makers as \emph{reciprocal} with respect to $\sigma_1$. The next definition provides the implicit relationship between posterior beliefs characterizing reciprocal decision-makers.

\begin{definition}\label{def:reciprocal}
Two decision-makers, $i$ and $j$, with $p^i < p^j$, are said to be ``reciprocal'' for $\sigma_1$ if,
after observing $\tilde\sigma_1=\sigma_1$, (i) both have prior beliefs that belong to the non-extreme set $\neg\mathcal{E}_{\sigma_1}$, and (ii) their posterior beliefs are related as follows,
\begin{equation}\label{eq:reciprocal}
p^i_{\sigma_1 \alpha} = \frac{1}{2} + \frac{P^j (\beta_2 \mid  \sigma_1
  )}{P^i (\alpha_2 \mid  \sigma_1 )}\left( \frac{1}{2} - p^j_{\sigma_1
    \beta} \right).
\end{equation}
\end{definition}
Equation~\eqref{eq:reciprocal} outlines a specific relationship between posterior beliefs that implicitly defines a condition on prior beliefs. Importantly, such a condition is conditional on the realization of the signal's first component. The next result shows that decision-makers that are reciprocal for a signal realization $\sigma_1'$ must share the same set of admissible cost conditional on $\sigma_1'$. At the same time, they cannot be reciprocal for the other realization $\sigma_1''$, on which they have different willingness to pay for acquiring the signal's second component.

\begin{lemma}\label{lemma:nonreciprocal}
    Two decision-makers that are reciprocal for $\sigma_1$ share the same set of admissible costs, $\mathcal{C}_{\sigma_1}(p^i)=\mathcal{C}_{\sigma_1}(p^j)$. If they are reciprocal for signal realization $\sigma_1'$, then they are not reciprocal for $\sigma_1''$, where $\sigma_1'\neq \sigma_1''$.
\end{lemma}

\begin{proof}
For the first part of the proof, rewrite equation~\eqref{eq:reciprocal} as $P^i (\alpha_2 \mid  \sigma_1
)\left[p^i_{\sigma_1 \alpha} - \frac{1}{2}\right] = P^j (\beta_2 \mid 
\sigma_1)\left[\frac{1}{2} - p^j_{\sigma_1 \beta}\right]$. Multiplying both sides
by $\Delta U$ we get,
\begin{equation*}
\Delta U \cdot P^i (\alpha_2 \mid  \sigma_1) \cdot \left[2p^i_{\sigma_1 \alpha} -
1\right] = \Delta U \cdot P^j(\beta_2 \mid  \sigma_1) \cdot \left[1 - 2p^j_{\sigma_1 \beta}\right].
\end{equation*}
Depending on the realization of $\tilde\sigma_1$, the left-hand side is DM~$i$'s cost function for cases 3 and 6. The right-hand side is DM~$j$'s cost function for cases 2 and 7. It follows that $i$ and $j$ share the same cost threshold and thus the set of admissible
cost conditional on $\sigma_1$.

The second part of the proof begins with the observation that the only prior for which $c_{\alpha}(p)=c_{\beta}(p)$ is $p=\nicefrac{1}{2}$. Reciprocal DMs have different prior beliefs, and thus at least one of them must have different cost thresholds for different realizations of $\tilde\sigma_1$. Their prior beliefs must belong to different cases (see Table~\ref{tab:cases:posterior}) because of the strict monotonicity of the cost function, $c_{\sigma_1}(p)$, in the domain of non-extreme sets, $\neg \mathcal{E}_{\sigma_1}$. Otherwise, they could not share the same set of admissible costs. From the cost function's shape (see Section~\ref{sec:incentives}) we have that, on their shared domain, $c^3_{\alpha}(\cdot)>c^6_{\beta}(\cdot)$, and $c^2_{\alpha}(\cdot)<c^7_{\beta}(\cdot)$. Hence, if $c_{\sigma_1'}\left(p^i\right)=c_{\sigma_1'}\left(p^j\right)$, then
$c_{\sigma_1''}\left(p^i\right) \neq c_{\sigma_1''}\left(p^j\right)$ for $p^i \neq p^j$ and $\sigma_1' \neq \sigma_1''$. The DMs do not share the same set of admissible costs, and thus they are not reciprocal for $\sigma_1''$.
\end{proof}

Only pairs of decision-makers who have the same prior beliefs (and thus cannot be reciprocal) share the same set of admissible costs regardless of $\sigma_1$, with the exception of the pathological case where $p^i=0$ and $p^j=1$. From Definition~\ref{def:reciprocal} and Lemma~\ref{lemma:nonreciprocal}, we obtain that $i$ has a higher willingness to pay for information than $j$ (i.e., $c_{\sigma_1}(p^i) > c_{\sigma_1}(p^j)$) if and only if $p^i_{\sigma_1 
  \alpha} < \frac{1}{2} + \frac{P^j (\beta_2 \mid  \sigma_1    )}{P^i (\alpha_2 \mid  \sigma_1 )}\left[ \frac{1}{2} - p^j_{\sigma_1 \beta}
\right]$. Given signal realization $\sigma_1$, and two different but non-reciprocal decision-makers, if at least one of them belongs to the non-extreme set ($\neg \mathcal{E}_{\sigma_1}$), then one decision-maker must be willing to pay a higher processing cost than the other. In these cases, there are processing costs for which only one decision-maker would pay to observe $\tilde\sigma_2$. Such a set would be empty for decision-makers who are either reciprocal or have identical priors. As we shall see, conditions under which only one decision-maker acquires additional information are key for the analysis that follows. The next definition formally defines the sets containing all pairs of prior beliefs for which, conditional on $\sigma_1$, decision-makers have different incentives for information acquisition.

\begin{definition}\label{def:sets_sigma2}
    Consider two decision-makers, $k$ and $l$. Given a processing cost $c$, the set of all pairs of prior beliefs such that, after observing $\sigma_1$, decision-maker $k$ chooses to observe $\sigma_2$ while decision-maker $l$ does not is
    \[
    \mathcal{B}^{kl}_{\sigma_1}(c) := \left\{ \left(p^i,p^j\right) \mid p^k \in \mathcal{H}_{\sigma_1}(c) \text{ and } p^l \notin \mathcal{H}_{\sigma_1}(c) \text{ for some } k,l\in\{i,j\} \right\}.
    \]
    The set of prior beliefs for which decision-maker $k$ observes $\sigma_2$ while decision-maker $l$ does not, after observing $\sigma_1$ and for some processing cost, is
    \[
    \mathcal{V}^{kl}_{\sigma_1} := \bigcup_{c>0} \mathcal{B}^{kl}_{\sigma_1}(c). 
    \]
\end{definition}

To clarify intuitions, it is useful to explore the connection between the set $\mathcal{V}^{kl}_{\sigma_1}$ and the decision-makers’ willingness to pay to observe the signal's second component, $\tilde\sigma_2$. From Definition \ref{def:sets_sigma2}, we see that  $\mathcal{V}^{kl}_{\sigma_1}$ contains all pairs of prior beliefs for which, at some processing cost, only decision-maker $k$ (optimally) observes $\sigma_2$. This occurs only when $k$ and $l$, given their prior beliefs and $\sigma_1$, have different willingness to pay for additional information. By definition, we can express  $\mathcal{V}^{kl}_{\sigma_1}$ alternatively as
\[
\mathcal{V}^{kl}_{\sigma_1} \triangleq \left\{ \left(p^k,p^l\right) \;\middle|\;  c_{\sigma_1}\left(p^k\right) > c_{\sigma_1}\left(p^l\right) \right\}.
\]
If two decision-makers’ prior beliefs belong to $\mathcal{V}^{kl}_{\sigma_1}$, at least one must hold non-extreme beliefs, and they must not be reciprocal for $\sigma_1$. As we shall see, the sets $\mathcal{V}^{kl}_{\sigma_1}$ and $\mathcal{B}^{kl}_{\sigma_1}(c)$ play an important role in the analysis of belief polarization. Moreover, there is a tight connection between the sets $\mathcal{V}_{\sigma_1}^{kl}$ and $\mathcal{B}_{\sigma_1}^{kl}(c)$. By definition, $(p^i,p^j)\in \mathcal{B}^{kl}_{\sigma_1}(c)$ if and only if $(p^i,p^j)\in \mathcal{V}^{kl}_{\sigma_1}$ and $c \in \mathcal{C}_{\sigma_1}(p^k) \setminus \mathcal{C}_{\sigma_1}(p^l)$.

\subsection{Belief polarization}\label{sec:polarization}

Polarization of beliefs occurs when two or more individuals, despite having access to the same information, interpret it in ways that drive their beliefs further apart. Moreover, instead of converging toward a common understanding, decision-makers update their prior beliefs in opposite directions, leading to posteriors that are more distant from each other than before. This phenomenon is particularly striking because it suggests that even shared evidence may reinforce pre-existing disagreement, deepening divisions rather than fostering consensus. This section provides a formal definition of \emph{polarized beliefs} within the framework under scrutiny. It then outlines the necessary and sufficient conditions under which beliefs polarize according to this definition.

There are two defining characteristics of polarized beliefs that we can formalize by using concepts introduced in Section~\ref{sec:intermediate}. First, decision-makers' beliefs grow further apart. Second, their beliefs move in opposite directions. The former is captured by the divergence function described by equation~\eqref{eq:div}, and the latter by the inversion function described by equation~\eqref{eq:inv}. Following Definition~\ref{def:daiu}, this paper defines polarization of beliefs as the simultaneous occurrence of diverging attitudes and inverse updating.\footnote{The notion of polarization here is more general than that used by, e.g., \cite{chang}. In~\citet[p.~3]{chang}, belief divergence refers to ``cases in which the person with the stronger belief in a hypothesis increases the strength of his or her belief, and the person with the weaker belief decreases the strength of his or her belief.'' Their definition excludes cases where beliefs change order. In contrast, the notion of belief divergence employed here (see Definition~\ref{def:daiu}) admits cases where beliefs swap order in the sense that if $p^i>p^j$, then $p^i(\sigma)<p^j(\sigma)$. Lemma~\ref{lemma:PB_no_swap} shows that polarization cannot happen via beliefs swap.}

\begin{definition}\label{def:pb}
Polarization of beliefs (PB) occurs when $I(\sigma)<0$ and $D(\sigma)<0$.
\end{definition}

The next proposition shows necessary and sufficient conditions for belief polarization based on the decision-makers' prior beliefs and the signal's structure. Specifically, it shows that polarization requires decision-makers, while receiving the same signal, $\sigma$, to scrutinize it differently: one must observe only its first component, while the other must observe it in its entirety. This requires at least one decision-maker to be willing, in principle, to incur some processing cost to observe the signal as a whole. A failure of any of these conditions implies that the decision-makers update their priors in the same direction. In such cases, their beliefs may diverge, but they cannot polarize.\footnote{Two decision-makers with identical prior beliefs share the same incentives and thus make the same information acquisition decisions. Consequently, polarization of beliefs cannot occur between them. However, in a variation of the model, polarization between decision-makers with identical priors could still occur if they have different payoff structures.}

\begin{proposition}\label{prop:PB_IFF}
    Consider two decision-makers, $i$ and $j$, such that $p^i<p^j$. Polarization of beliefs between decision-makers $i$ and $j$ occurs with positive ex-ante probability and for some processing cost if and only if,
    \begin{itemize}[noitemsep]
        \item[i)] $\left(p^i,p^j\right)\in \mathcal{V}^{ij}_{\alpha} \cup \mathcal{V}^{ji}_{\beta}$;
        \item[ii)] $\theta_2>\theta_1$.
    \end{itemize}
    Given a processing cost $c>0$, polarization of beliefs occurs with positive ex-ante probability if and only if
    \begin{itemize}[noitemsep]
    \item[ii)] $\theta_2>\theta_1$;
    \item[iii)] $c<\max\{c_{\sigma_1}(p^i),c_{\sigma_1}(p^j)\}$;
     \item[iv)] $(p^i,p^j)\in\mathcal{B}^{ij}_{\alpha}(c)\cup\mathcal{B}^{ji}_{\beta}(c)$.
    \end{itemize}
\end{proposition}

Proposition~\ref{prop:PB_IFF} establishes the necessary and sufficient conditions for belief polarization to occur with some positive ex-ante probability. However, it does not specify the events that trigger polarization or quantify its likelihood. The next proposition addresses this gap, demonstrating that polarization arises when the signal's components differ, i.e., $\sigma_1 \neq \sigma_2$. In other words, a superficial reading of $\tilde\sigma$ must support a different state than its more precise scrutiny. The proposition also provides the probability of this event, showing that it occurs less than half the time. Thus, polarization is inherently less likely to occur than to not occur.

\begin{proposition}\label{prop:probability_PB}
    When conditions i) to iv) in Proposition \ref{prop:PB_IFF} are satisfied, and given a subjective probability that the state is $s=A$ equal to $p$, polarization of beliefs occurs with an ex-ante probability given by $Pr(PB)$, where
    \begin{multline}\label{eq:probPB}
    Pr(PB)=Pr(\tilde\sigma_1=\alpha)\cdot Pr(\tilde\sigma_2 = \beta)\cdot \mathds{1}\left\{ (p^i,p^j) \in\mathcal{B}^{ij}_{\alpha}(c) \right\}\\
    + Pr(\tilde\sigma_1=\beta)\cdot Pr(\tilde\sigma_2 = \alpha)\cdot \mathds{1}\left\{ (p^i,p^j) \in\mathcal{B}^{ji}_{\beta}(c)\right\}.
    \end{multline}
When any of the conditions $i)$ to $iv)$ is not satisfied, then $Pr(PB)=0$. Polarization of beliefs occurs with an ex-ante probability that can be up to but no larger than $\nicefrac{1}{2}$.
\end{proposition}

Proposition~\ref{prop:probability_PB} shows that the ex-ante probability of polarization depends on the subjective prior beliefs, which remain indeterminate due to the model's allowance for different beliefs about the state. Consequently, decision-makers and external observers with different priors would also disagree on the likelihood of polarization. The next corollary establishes that, provided the signal's second component is more informative than the first, (almost) no decision-maker is immune to polarization.

\begin{corollary}\label{cor:polarizing} 
For any decision-maker~$i$ with prior $p^i \in (0,1)$, and given some processing cost $c \in \left(0, \Delta U \left[\theta_2 - \frac{1}{2}\right] \right)$, there exists a non-empty set of priors, $\mathcal{T}\subset[0,1]$, such that, given any decision-maker~$j$ with prior $p^j \in \mathcal{T}$, the condition $\theta_2 > \theta_1$ is sufficient for belief polarization to occur between $i$ and $j$ with positive ex-ante probability. 
\end{corollary}

\begin{figure}
\centering
\begin{subfigure}{.45\textwidth}
  \centering
  \includegraphics[width=.9\linewidth]{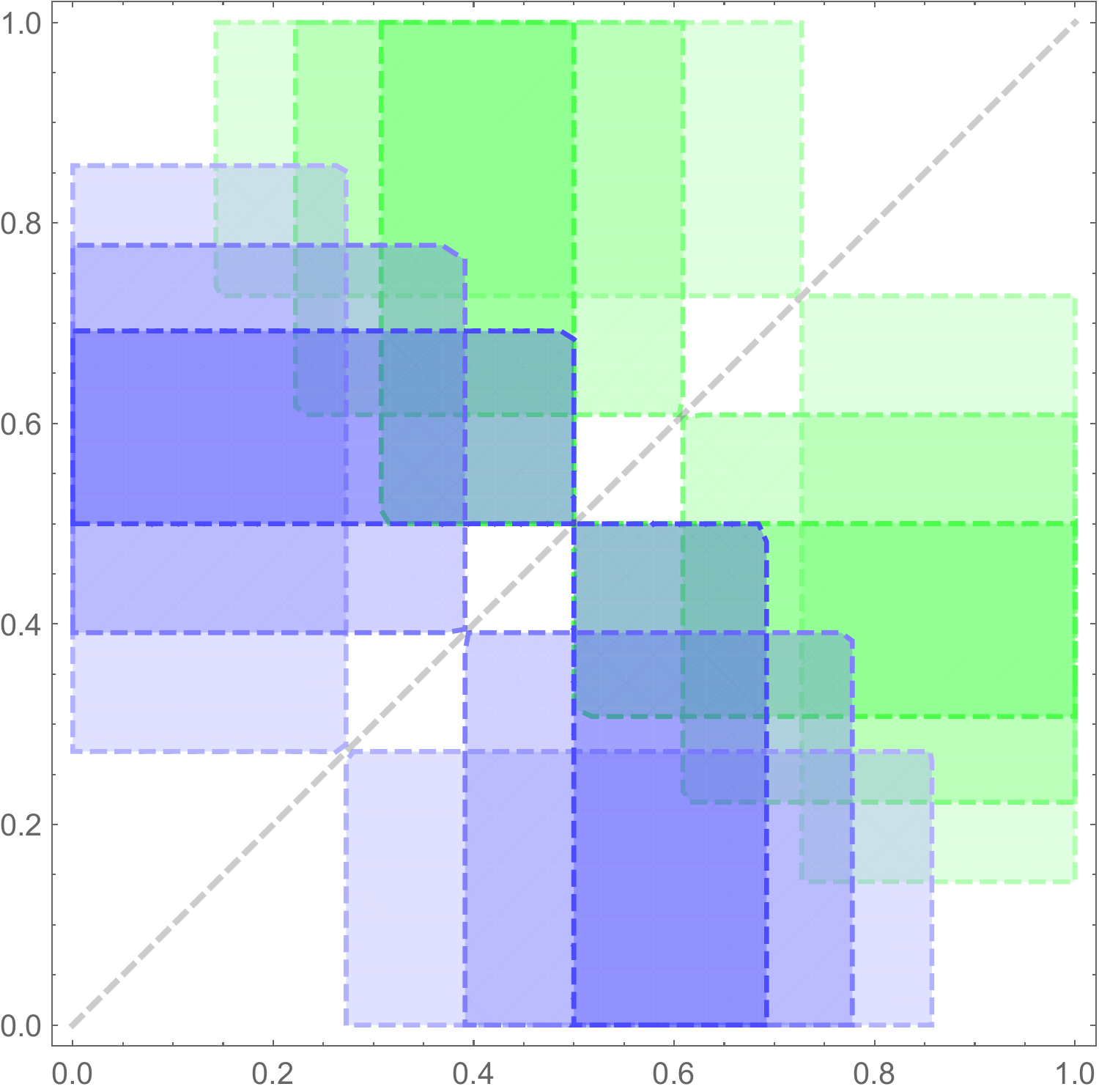}
  \caption{$\mathcal{B}^{ij}_{\alpha}(c)$ and $\mathcal{B}^{ji}_{\beta}(c)$}
  \label{fig:setsJK}
\end{subfigure}%
\begin{subfigure}{.45\textwidth}
  \centering
    \includegraphics[width=.9\linewidth]{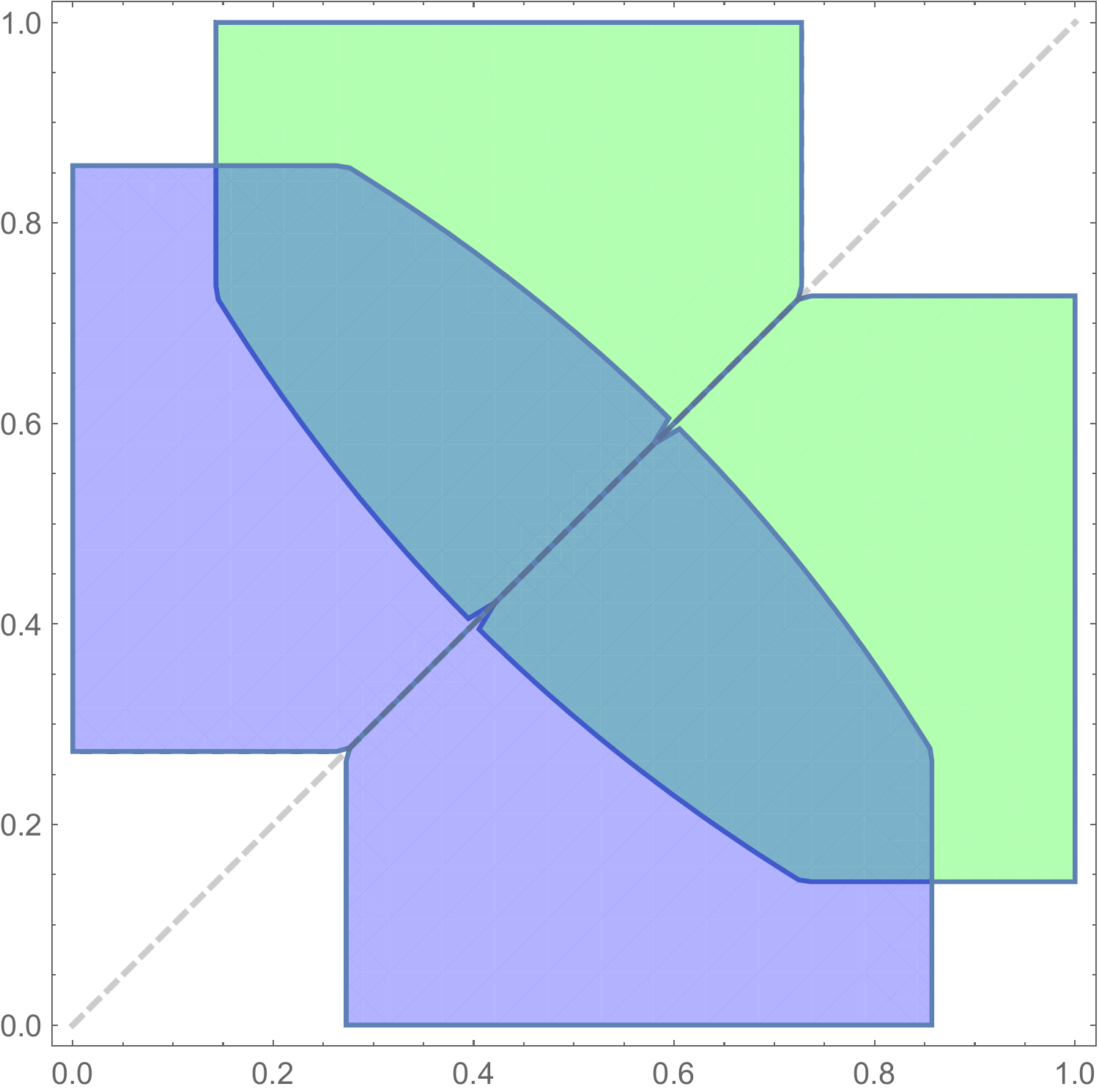}
  \caption{$\mathcal{V}^{ij}_{\alpha}$ and $\mathcal{V}^{ji}_{\beta}$}
  \label{fig:setsH}
\end{subfigure}
\caption{In both panels, the decision-makers' prior beliefs are represented on the axes. The left panel displays the sets $\mathcal{B}^{ij}_{\alpha}(c)$ and $\mathcal{B}^{ji}_{\beta}(c)$, while the right panel illustrates the sets $\mathcal{V}^{ij}_{\alpha}$ and $\mathcal{V}^{ji}_{\beta}$. Cases where $\sigma_1 = \alpha$ are shown in green, and cases where $\sigma_1 = \beta$ are in blue. In the left panel, a darker shade denotes a higher processing cost, with $c \in \{0, 0.1, 0.2\}$. The dashed 45-degree line indicates $p^i = p^j$. Parameters are set to $\theta_1 = \nicefrac{3}{5}$, $\theta_2 = \nicefrac{4}{5}$, and $\Delta U = 1$.}
\label{fig:sets}
\end{figure}

\subsection{Confirmatory belief patterns and reaction to information}\label{sec:behaviors}

This section examines the belief patterns of a single decision-maker. The model is grounded in standard Bayesian inference and rational choice, making it inappropriate to denote the decision-maker's optimal behavior or belief updating as \emph{biased}. However, an external observer might perceive them as ``\emph{as if}'' biased when assuming that the decision-maker fully knows and processes $\sigma$, thus neglecting the potential lack of incentives for information acquisition. Under this inaccurate lens, the decision-maker may make decisions and have beliefs that are reminiscent of several well-know biases in information processing, while adhering to the model's rational framework.

The next definition provides a formalization for attitudes that are akin to what some authors refer to as \emph{disconfirmation bias}. This bias is roughly described as ``a propensity to [...] accept confirming evidence at face value while scrutinizing disconfirming evidence hypercritically'' \citep[p.~2099]{lord} or when ``arguments incompatible with prior beliefs are scrutinized longer, subjected to more extensive refutational analyses'' \citep{edwards}. Definition~\ref{def:disconfirmatory} introduces the concept of \emph{disconfirmation}, which pertains not to belief patterns but to the incentives for and choices of information acquisition. 

\begin{definition}\label{def:disconfirmatory}
A decision-maker with prior $p>\nicefrac{1}{2}$ has a tendency for disconfirmation when $c_\beta(p)>c_\alpha(p)$, and exhibits disconfirmation when $a^*(\beta_1)=\rho \neq a^*(\alpha_1)=\neg\rho$. A decision-maker with prior $p<\nicefrac{1}{2}$ has a tendency for disconfirmation when $c_\alpha(p)>c_\beta(p)$, and exhibits disconfirmation when $a^*(\alpha_1)=\rho \neq a^*(\beta_1)=\neg\rho$.
\end{definition}

Decision-makers with non-extreme beliefs always exhibit a tendency for disconfirmation. As shown in Proposition~\ref{prop:cost_properties}, their willingness to pay for additional information is generally higher after encountering evidence that contradicts their prior beliefs. For example, if $p > \nicefrac{1}{2}$ and $p \in \neg \mathcal{E}_{\beta}$ -- indicating that the decision-maker's prior is non-extreme and favors $s = A$ -- then $c_\beta(p) > c_\alpha(p)$. This implies that the decision-maker’s willingness to pay for observing $\sigma_2$ is higher after receiving evidence supporting the state $s = B$ than after receiving a signal supporting $s = A$. Likewise, if $p < \nicefrac{1}{2}$ and $p \in \neg \mathcal{E}_\alpha$, then $c_\alpha(p) > c_\beta(p)$. The propensity to scrutinize information that contradicts one’s preconceptions is a direct implication of Bayesian updating and rational decision-making.\footnote{The decision-maker's willingness to pay for additional information at $p = \nicefrac{1}{2}$ is $c_\alpha(\nicefrac{1}{2}) = c_\beta(\nicefrac{1}{2}) = \max\{0, \Delta U[\theta_2 - \theta_1]\}$, which is positive if and only if $\theta_2 > \theta_1$.} Whether a decision-maker exhibits disconfirmation depends on the processing cost, which must fall within an intermediate range. The next proposition formalizes these results.

\begin{proposition}\label{prop:disconfirmation}
    A decision-maker with prior $p$ has a tendency for disconfirmation if and only if $p\in \neg\mathcal{E}\setminus\left\{\frac{1}{2}\right\}$. A decision-maker with prior $p>\nicefrac{1}{2}$ (resp. $p<\nicefrac{1}{2}$) exhibits disconfirmation if and only if the processing cost is such that $c_\beta(p)>c>c_\alpha(p)$ (resp. $c_\alpha(p)>c>c_\beta(p)$). This is possible only if $p\in\neg\mathcal{E}_\beta$ (resp. $p\in\neg\mathcal{E}_\alpha$). 
\end{proposition}

The following definition formalizes belief patterns that resemble what some authors describe as confirmation bias. Broadly, confirmation bias refers to the tendency to process information in a manner that reinforces one's existing beliefs or hypotheses. As \cite{nickerson} describes, it involves ``interpreting evidence in ways that are partial to existing beliefs.'' Definition~\ref{def:cb} introduces the concepts of \emph{confirmatory} and \emph{disproving} belief patterns. While confirmatory patterns may resemble attitudes associated with confirmation bias, disproving patterns are introduced to complement the analysis.

\begin{definition}\label{def:cb}
A decision-maker with prior beliefs $p>\nicefrac{1}{2}$ (resp. $p<\nicefrac{1}{2}$) exhibits confirmatory belief patterns (CB) when ${p}_{\sigma}<p<{p}(\sigma)$ (resp. ${p}(\sigma)<p<{p}_{\sigma}$), and exhibits disproving belief patterns (DB) when
${p}(\sigma)<p<{p}_{\sigma}$ (resp. ${p}_{\sigma}<p<{p}(\sigma)$). 
\end{definition}

When the cost of processing additional information, $c$, is prohibitively high, the decision-maker never observes the signal's second component. In this scenario, both confirmatory and disproving belief patterns can arise. Conversely, when $c$ is very low, the decision-maker always scrutinizes the second component, leaving no room for either pattern to occur. However, within a moderate range of $c$, the decision-maker observes the second component selectively---only when the first component presents contradictory evidence. In this range, confirmatory belief patterns can emerge, while disproving patterns cannot. Interestingly, this is also the range where the decision-maker has a tendency for disconfirmation. 

When confirmatory evidence discourages further scrutiny, the decision-maker remains unaware of the second component's realization. If the second component, when unobserved, is both contradictory and more informative, the decision-maker exhibits confirmatory belief patterns. This outcome underscores an intriguing connection: confirmatory belief patterns often coexist with a tendency for disconfirmation. The next proposition provides necessary and sufficient conditions for confirmatory and disproving belief patterns.

\begin{proposition}\label{prop:confirmation}
Necessary and sufficient conditions for confirmatory and disproving belief patterns in a decision-maker with prior $p$, are
\begin{itemize}[noitemsep]
    \item a relatively more informative second component of the signal, i.e., $\theta_2 > \theta_1$;
    \item a sufficiently high processing cost, i.e., $p \notin \mathcal{H}_{\sigma_1}(c)$ or, equivalently, $c > c_{\sigma_1}(p)$;
    \item for CB: if $p > \nicefrac{1}{2}$, then $\sigma = (\alpha, \beta)$; if $p < \nicefrac{1}{2}$, then $\sigma = (\beta, \alpha)$;
    \item for DB: if $p > \nicefrac{1}{2}$, then $\sigma = (\beta, \alpha)$; if $p < \nicefrac{1}{2}$, then $\sigma = (\alpha, \beta)$.
\end{itemize}
\end{proposition}

The next definition introduces belief patterns describing individuals’ reactions to information, as seen in how they update beliefs based on new evidence. Underreaction occurs when an individual adjusts their beliefs less than expected given the evidence. Overreaction, happens when beliefs shift more than warranted by the evidence. ``The experimental evidence on inference taken as a whole suggests that [...] people generally underinfer rather than overinfer'' \citep{benjamin2019errors}.

In this rational Bayesian model, decision-makers cannot inherently under- or overreact when processing all available information, as their belief updates follow Bayes' rule. However, underreaction or overreaction can arise relative to the full information set if decision-makers only acquire part of the available information---specifically, when they observe only the first component of the signal, $\sigma_1$, but not the second, $\sigma_2$. This partial information processing leads to deviations in belief updating compared to the scenario where the complete signal $\sigma = (\sigma_1, \sigma_2)$ is observed. The following definition formalizes the notions of overreaction and underreaction with respect to $\sigma$.

\begin{definition}\label{def:or_ur}
    A decision-maker underreacts to $\sigma$ (UR) if either $p < {p}(\sigma) < p_{\sigma}$ or $p_{\sigma} < p(\sigma) < p$; she overreacts to $\sigma$ (OR) if either $p < p_{\sigma} < p(\sigma)$ or $p(\sigma) < p_{\sigma} < p$.
\end{definition}

A decision-maker can under- or overreact to $\sigma$ only if they choose not to observe the signal's second component. This outcome requires the processing cost to be sufficiently high to discourage the acquisition and processing of additional information. In such cases, the direction of the reaction depends on whether the components of the signal support the same state or not. Furthermore, overreaction to $\sigma$ is possible only if the second component of the signal is less informative than the first. The following proposition formalizes these results, and concludes the analysis of belief patterns.

\begin{proposition}\label{prop:reaction}
    A decision-maker underreacts to $\sigma$ if and only if $\sigma_1 = \sigma_2$ and $c > c_{\sigma_1}(p)$. A decision-maker overreacts to $\sigma$ if and only if $\sigma_1 \neq \sigma_2$, $c > c_{\sigma_1}(p)$, and $\theta_2 < \theta_1$.
\end{proposition}

\section{Concluding remarks}

This paper presents a model of costly information acquisition that explores how decision-makers process and act on information under rational Bayesian principles. By introducing a framework where individuals can choose between superficial and precise scrutiny of evidence, the analysis highlights how belief patterns such as polarization, confirmation bias, and underreaction to information may seem to emerge, even among rational agents. These patterns are not indicative of cognitive biases but are instead driven by differing incentives for information acquisition. When these incentives are neglected, the decision-makers' optimal behavior is reminiscent of several biases in judgment and decision-making investigated in both psychology and economics. The findings suggest that phenomena often attributed to biases in human information processing may not stem from flaws in the decision-makers themselves but rather from observers neglecting to account for the role of incentives in information acquisition.

In this framework, belief polarization arises when individuals with different prior beliefs make distinct choices about acquiring additional information. While some choose to scrutinize evidence fully, others opt for partial processing due to cost constraints. These differing courses of action, combined with the structure of the signal, may lead to diverging beliefs, even when individuals share access to the same evidence. Importantly, polarization is shown to occur only under specific conditions, such as when one decision-maker observes only part of the signal and when the signal's components provide conflicting indications. A theoretical bound is provided: within this model, belief polarization should not be expected to occur more often than not. 

Beyond polarization, the model also sheds light on individual belief patterns. Decision-makers may exhibit behaviors reminiscent of confirmation bias or underreaction to information when they process evidence selectively. However, these outcomes are fully consistent with rational decision-making in the presence of processing costs. This rational framework provides a nuanced perspective, suggesting that behaviors often attributed to cognitive biases can result from optimal responses to constraints.

These results have significant implications for understanding information processing and belief formation in economics and psychology. The model offers testable predictions about how belief patterns vary with prior beliefs, the informativeness of signals, and the cost of processing information. Future research can build on these insights to further investigate the interplay between incentives, information structures, and belief dynamics in various decision-making contexts. An empirical investigation may prove fruitful in testing these implications and assessing the extent to which individuals deviate from a benchmark that explicitly accounts for decision-makers' incentives in information processing.

\clearpage

\appendix

\section{Appendix}

\subsection{Bayesian updating}\label{sec:posteriors}

Consider a DM with prior $p \in [0,1]$. The posterior beliefs after observing $\sigma_1= \alpha$ are
\begin{equation*}
p_{\alpha}=P(A\mid \alpha_1) = \frac{\theta_1 p}{\theta_1 p + (1- \theta_1)(1-p)}.
\end{equation*}
Conversely, the posterior beliefs after observing $\sigma_1= \beta$ are
\begin{equation*}
p_{\beta}=P(A\mid \beta_1) = \frac{(1 - \theta_1) p}{(1- \theta_1) p +
\theta_1 (1-p)}.
\end{equation*}
Suppose that the DM observes also the realization of $\tilde\sigma_2$. The posterior beliefs after observing $\sigma=(\sigma_1,\sigma_2)$ are
\begin{equation*}
p_{\alpha \alpha}=P\left(A\mid \alpha_1, \alpha_2\right) = \frac{\theta_1 \theta_2 p}{\theta_1 \theta_2 p + (1- \theta_1)(1-\theta_2)(1-p)},
\end{equation*}
\begin{equation*}
p_{\alpha \beta}=P\left(A\mid \alpha_1, \beta_2\right) = \frac{\theta_1 (1-\theta_2)p}{\theta_1 (1- \theta_2) p +
(1- \theta_1) \theta_2(1-p)},
\end{equation*}
\begin{equation*}
p_{\beta \beta}=P\left(A\mid \beta_1, \beta_2\right) = \frac{(1- \theta_1) (1- \theta_2) p}{(1- \theta_1) (1- \theta_2) p +
\theta_1 \theta_2(1-p)},
\end{equation*}
\begin{equation*}
p_{\beta \alpha}=P\left(A\mid \beta_1, \alpha_2\right) = \frac{(1-\theta_1) \theta_2 p}{(1 - \theta_1) \theta_2 p + \theta_1
  (1- \theta_2)(1-p)}.
\end{equation*}

The prior can also be represented as a function of the posteriors, and similarly, the posteriors after observing $\sigma_1$ (i.e., $p_{\sigma_1}$) can be expressed as a function of the posteriors after observing $\sigma_2$ (i.e., $p_{\sigma}$). This is because the prior is a convex combination of the posteriors. Within this alternative representation, we have
\begin{eqnarray*}
p & = & p_{\alpha} P(\sigma_1= \alpha)+p_{\beta} P(\sigma_1 = \beta)
 \\
&=& p_{\alpha} [P(\alpha_1 \mid  A) P(A) + P(\alpha_1 \mid 
B)P(B)] \\
 & & + \: p_{\beta}
[P(\beta_1 \mid  A) P(A) + P(\beta_1 \mid  B)P(B)].
\end{eqnarray*}
Likewise, the interim posteriors $p_{\alpha}$ and $p_{\beta}$ can be represented as
convex combinations of posteriors $p_{\sigma}$,
\begin{eqnarray}
p_{\alpha} & = & p_{\alpha \alpha} P(\alpha_2 \mid \alpha_1)+p_{\alpha \beta} P(\beta_2 \mid  \alpha_1),
\label{palfa}  \\
p_{\beta} & = & p_{\beta \alpha} P(\alpha_2 \mid \beta_1)+p_{\beta \beta} P(\beta_2 \mid  \beta_1).
\label{pbeta}  
\end{eqnarray}
The unconditional signal probability distributions $P(\alpha_1)$ and $P(\beta_1)$ are
\[
P(\alpha_1)= p\theta_1 + (1-p)(1-\theta_1),
\]
\[
P(\beta_1)=p(1-\theta_1)+(1-p)\theta_1,
\]
whereas the conditional signal probability distributions are
\begin{equation}\label{eq:cond_alpha}
P(\alpha_2 \mid  \sigma_1) = p_{\sigma_1}\theta_2 + (1-p_{\sigma_1})(1-\theta_2),
\end{equation}
\begin{equation}\label{eq:cond_beta}
P(\beta_2 \mid  \sigma_1) = p_{\sigma_1}(1-\theta_2) + (1-p_{\sigma_1})\theta_2.
\end{equation}

\subsection{Prior categorization and proof of Proposition~\ref{prop:cost_function}}\label{sec:cases}

This section offers a rationale for the categorization outlined in Section~\ref{sec:differentcases} and provides a proof of Proposition~\ref{prop:cost_function}, which formally defines the cost function.

\subsubsection{Cases 1 to 4}

\subsubsection*{Case 1: $p_{\alpha} \geq \frac{1}{2}$, $p_{\alpha \alpha} \geq
  \frac{1}{2}$, $p_{\alpha \beta} \geq \frac{1}{2}$}

The condition $p_{\alpha \beta} \geq \nicefrac{1}{2}$ is necessary for a decision-maker who has observed $\sigma_1 = \alpha$ to belong to case~1, and it implies that $p_{\alpha \alpha}$ is strictly greater than $\nicefrac{1}{2}$. By applying Bayes' rule, we find that $p_{\alpha \beta} \geq \nicefrac{1}{2}$ implies $p \geq \frac{\theta_2 (1 - \theta_1)}{\theta_1 + \theta_2 - 2 \theta_1 \theta_2}$. We further require $p_\alpha \geq \nicefrac{1}{2}$, which implies $p>1-\theta_1$. This last condition is always satisfied when $p_{\alpha\beta}\geq \nicefrac{1}{2}$. It follows that a decision-maker belongs to case~1 only if her prior beliefs satisfy\footnote{The threshold $\frac{\theta_2 (1- \theta_1)}{\theta_1 + \theta_2 - 2 \theta_1  \theta_2} \in \left[ \frac{1}{2}, \theta_1 \right]$ when $\theta_2 \leq\frac{\theta_{1}^{2}}{1-2\theta_1 + 2\theta_{1}^{2}}$, and $\frac{\theta_2 (1- \theta_1)}{\theta_1 + \theta_2 - 2 \theta_1  \theta_2} \in (
\theta_1 ,\theta_2 ]$ otherwise.}
\begin{displaymath}
p \in \left[\frac{\theta_2 (1- \theta_1)}{\theta_1 + \theta_2 - 2 \theta_1  \theta_2},1 \right].
\end{displaymath}

After observing $\sigma_1 = \alpha$, the decision-maker has posterior beliefs $p_\alpha>\nicefrac{1}{2}$. Suppose that she takes a decision without observing the signal's second component, $\sigma_2$. Because the posterior exceeds $\nicefrac{1}{2}$, her optimal choice is $s^*=A$. This course of action grants the decision-maker an expected utility of
\begin{equation}\label{eq:case1_nopay}
\mathbb{E}\left[u\left(\neg\rho,A,\omega\right)\mid \alpha_1\right]= p_{\alpha}\overline{U} + (1-p_{\alpha}) \underline{U}.
\end{equation}
Suppose now that the decision-makers incurs the processing cost to scrutinize $\sigma_2$ before making a decision. By doing so, she anticipates that her posterior beliefs will be $p_{\alpha \sigma_2}\geq\nicefrac{1}{2}$ with probability $P(\sigma_2 \mid  \alpha_1)$. This course of action grants her an expected utility of
\begin{eqnarray}\label{eq:case1_pay}
\mathbb{E}\left[u\left(\rho,A,\omega\right)\mid \alpha_1, \tilde\sigma_2\right] &=&  \left[ P(\alpha_2 \mid  \alpha_1)p_{\alpha \alpha} + P(\beta_2 \mid  \alpha_1)p_{\alpha \beta} \right]\overline{U} \nonumber \\
 & & + \:  \left\{
  P(\alpha_2 \mid  \alpha_1)[1- p_{\alpha \alpha}] + P(\beta_2 \mid 
  \alpha_1)[1- p_{\alpha \beta}] \right\}\underline{U} -c.
\end{eqnarray}

The decision-maker is willing to incur the processing cost if and only if \eqref{eq:case1_pay} is greater than \eqref{eq:case1_nopay}. The condition on processing costs is
\begin{eqnarray}\label{eq:case1_cost}
c & \leq &  [
  P(\alpha_2 \mid  \alpha_1)p_{\alpha \alpha} + P(\beta_2 \mid 
  \alpha_1)p_{\alpha \beta} - p_{\alpha} ]\overline{U} \nonumber \\
 & & + \: \{
  P(\alpha_2 \mid  \alpha_1)[1- p_{\alpha \alpha}] + P(\beta_2 \mid 
  \alpha_1)(1- p_{\alpha \beta}) - [1-p_{\alpha}] \}\underline{U}.
\end{eqnarray}
The posterior $p_\alpha$ can be represented as a convex combination of $p_{\alpha \alpha}$ and $p_{\alpha \beta}$, that is,
\begin{equation}\label{eq:case1_convex}
p_{\alpha}  =  p_{\alpha \alpha} P(\alpha_2 \mid  
\alpha_1)+p_{\alpha \beta} P(\beta_2 \mid  \alpha_1). 
\end{equation}
By replacing \eqref{eq:case1_convex} in \eqref{eq:case1_cost}, we obtain
\begin{eqnarray}\label{eq:case1_condition}
c & \leq & \overline{U} \{
  P(\alpha_2 \mid  \alpha_1)p_{\alpha \alpha} + P(\beta_2 \mid 
  \alpha_1)p_{\alpha \beta} - [ P(\alpha_2 \mid  \alpha_1) p_{\alpha
    \alpha}+ P(\beta_2 \mid  \alpha_1) p_{\alpha \beta}] \} \nonumber \\ 
 & & + \: \underline{U} \{
  P(\alpha_2 \mid  \alpha_1)[1- p_{\alpha \alpha}] + P(\beta_2 \mid 
  \alpha_1)[1- p_{\alpha \beta}] \nonumber \\
& & - \: 1+ P(\alpha_2 \mid 
  \alpha_1) p_{\alpha \alpha}+ P(\beta_2 \mid  \alpha_1) p_{\alpha \beta} \} = 0, 
\end{eqnarray}
which is never satisfied because it is assumed that $c>0$. As a result, decision-makers belonging to case~1 never incur a cost to scrutinize the signal $\sigma$ as a whole. This outcome is straightforward given the premises because, in this case, no realization of $\tilde\sigma_2$ can alter the decision-maker's optimal action $s^*$. Nevertheless, the above procedure---first identifying the prior interval, and then deriving the condition on costs, as in \eqref{eq:case1_condition},---is instructive for understanding the process used in the remaining cases.


\subsubsection*{Case 2: $p_{\alpha} \geq \frac{1}{2}$, $p_{\alpha \alpha} \geq
  \frac{1}{2}$, $p_{\alpha \beta} \leq \frac{1}{2}$}

There are two necessary conditions on the posteriors for an agent to
be in case 2: $p_{\alpha} \geq
\frac{1}{2}$ and $p_{\alpha \beta} \leq \frac{1}{2}$. The former
implies $p
\geq 1- \theta_1$ and $p_{\alpha \alpha} \geq
\frac{1}{2}$, while the latter implies $p \leq \frac{\theta_2 (1- \theta_1)}{\theta_1 + \theta_2 - 2 \theta_1
  \theta_2}$. For this case to be feasible, we need that $\frac{\theta_2 (1- \theta_1)}{\theta_1 + \theta_2 - 2 \theta_1
  \theta_2} \geq 1 - \theta_1$, which is always true whenever $\theta_2 \geq \frac{1}{2}$. To belong in case 2, the DM should have a prior $p$ such that
\begin{equation*}
p \in \left[  1-\theta_1 , \frac{\theta_2 (1- \theta_1)}{\theta_1 + \theta_2 - 2 \theta_1
  \theta_2} \right].
\end{equation*}
Upon paying the processing cost $c$ in order to scrutinize $\sigma_2$, the DM expects to receive,
\begin{eqnarray*}
\mathbb{E}[u(\cdot)\mid \alpha_1 , \tilde\sigma_2] &=& \overline{U} \{
  P(\alpha_2 \mid  \alpha_1)p_{\alpha \alpha} + P(\beta_2 \mid 
  \alpha_1)[1- p_{\alpha \beta}] \}  \\
 & & + \: \underline{U} \{
  P(\alpha_2 \mid  \alpha_1)[1- p_{\alpha \alpha}] + P(\beta_2 \mid 
  \alpha_1)p_{\alpha \beta}\} -c. 
\end{eqnarray*}
Since $\mathbb{E}[u(\cdot)\mid \alpha_1]= p_{\alpha}\overline{U} + (1-p_{\alpha}) \underline{U}$, we obtain the following
condition,
\begin{eqnarray*}
c & \leq & \overline{U} \{
  P(\alpha_2 \mid  \alpha_1)p_{\alpha \alpha} + P(\beta_2 \mid 
  \alpha_1)[1- p_{\alpha \beta}] - p_{\alpha} \}  \\
 & & + \: \underline{U} [
  P(\alpha_2 \mid  \alpha_1)(1- p_{\alpha \alpha}) + P(\beta_2 \mid 
  \alpha_1)p_{\alpha \beta} - (1-p_{\alpha})] .
\end{eqnarray*}
By substituting (\ref{palfa}) for $p_{\alpha}$, we obtain 
\begin{eqnarray*}
c & \leq & \overline{U} \{
  P(\alpha_2 \mid  \alpha_1)p_{\alpha \alpha} + P(\beta_2 \mid 
  \alpha_1)[1- p_{\alpha \beta}] - P(\alpha_2 \mid  \alpha_1) p_{\alpha
    \alpha} - P(\beta_2 \mid  \alpha_1) p_{\alpha \beta} \}  \\ 
 & & + \: \underline{U} \{
  P(\alpha_2 \mid  \alpha_1)[1- p_{\alpha \alpha}] + P(\beta_2 \mid 
  \alpha_1) p_{\alpha \beta} - 1+ P(\alpha_2 \mid 
  \alpha_1) p_{\alpha \alpha}+ P(\beta_2 \mid  \alpha_1) p_{\alpha \beta} \}.
\end{eqnarray*}
After some rearrangements and simplifications we obtain the condition
\begin{equation}
c  \leq \left[\overline{U} - \underline{U}\right] \cdot [P(\beta_2 \mid  \alpha_1)]
\cdot [1-2p_{\alpha \beta}].
\end{equation}
Therefore, the set of admissible costs is $\mathcal{C}_{\alpha}(p)=\left[ 0, \Delta U \cdot [P(\beta_2 \mid  \alpha_1)] \cdot [1-2p_{\alpha \beta}] \right]$. By assumption, $\overline{U} - \underline{U}>0$,
$P(\beta_2 \mid  \alpha_1)>0$, and, since in this case $p_{\alpha   \beta}\leq\frac{1}{2}$, $1-2p_{\alpha \beta}\geq0$. Thus, the cost threshold $\Delta U \cdot [P(\beta_2 \mid  \alpha_1)] \cdot [1-2p_{\alpha \beta}]$ is positive. 

The cost function, $c_\alpha(p)$, depends on the utility $u(\cdot)$, the probability that the realization of $\tilde\sigma_2$ is such that the optimal choice $s^*$ changes (e.g., $\tilde\sigma_2=\beta$ would switch $s^*=A$ to $s^*=B$, while $\tilde\sigma_2=\alpha$ would not change the optimal action) and finally on the posterior $p_{\alpha \beta}$. In particular the threshold is increasing on the premium for a correct guess, on the probability of observing a signal that
changes the optimal action and negatively on $p_{\alpha \beta}$. Given
$P(\beta_2 \mid  
\alpha_1)$ as in \eqref{eq:cond_alpha} and $c_\alpha(p)  =  \Delta U \cdot P(\beta_2 \mid \alpha_1) \cdot [1-2p_{\alpha \beta}]$, we have that
\begin{eqnarray*}
c_\alpha(p)& = & \Delta U \cdot
[(1-\theta_2)p_{\alpha} + \theta_2(1-p_{\alpha})]
\cdot [1-2p_{\alpha \beta}] \\ 
 & =  & \Delta U \cdot
\left[ \frac{\theta_1p - \theta_2p + \theta_2 - \theta_1
    \theta_2}{\theta_1 p + (1- \theta_1)(1-p)} \right]
\cdot \left[ \frac{\theta_2 (1-p) - \theta_1 p -    \theta_1 \theta_2 (1-2p)}{\theta_1p - \theta_2p + \theta_2 - \theta_1 \theta_2} \right].  
\end{eqnarray*}
Therefore, a DM with prior $p$ that belongs to case~2 has the following
cost function,
\begin{equation*}
c_\alpha(p)  =  \Delta U \cdot \left[ \frac{\theta_2 (1-p) - \theta_1 p -
    \theta_1 \theta_2 (1-2p)}{\theta_1 p + (1-\theta_1)(1-p)} \right]. 
\end{equation*}
The right hand side above is non-negative provided that $p \leq \frac{\theta_2
  (1-\theta_1)}{\theta_1 + \theta_2 - \theta_1 \theta_2}$.

\subsubsection*{Case 3: $p_{\alpha} \leq \frac{1}{2}$, $p_{\alpha \alpha} \geq
  \frac{1}{2}$, $p_{\alpha \beta} \leq \frac{1}{2}$}
  
In this case, $p_{\alpha} \leq \frac{1}{2}$ also implies that $p_{\alpha \beta}\leq \frac{1}{2}$. We further need $p_{\alpha \alpha} \geq \frac{1}{2}$. Those conditions together imply that the DM's prior must satisfy $p \leq 1 -\theta_1$ and $p \geq 1 - \frac{\theta_1 \theta_2 }{1- \theta_1 -
  \theta_2 + 2 \theta_1  \theta_2}$. For this interval to exist, we need, $1 - \frac{\theta_1 \theta_2 }{1- \theta_1 - \theta_2 + 2 \theta_1
  \theta_2} \leq 1- \theta_1$, which is true whenever $\theta_2 \geq
\frac{1}{2}$, thus always. Moreover, for $\theta_2<1$ we have that $1 - \frac{\theta_1 \theta_2 }{1- \theta_1 - \theta_2 + 2 \theta_1
  \theta_2} < 1- \theta_2$. It follows that a DM belongs to case 3 if her prior $p$ satisfies
\begin{equation*}
p \in \left[  1 - \frac{\theta_1 \theta_2 }{1- \theta_1 - \theta_2 + 2 \theta_1
  \theta_2} , 1- \theta_1 \right].
\end{equation*}
In this case, the DM's expected utility conditional on observing only $\sigma_1$ is
\begin{displaymath}
\mathbb{E}[u(\cdot)\mid \alpha_1]= (1 - p_{\alpha})\overline{U} + p_{\alpha}\underline{U}.
\end{displaymath}
If the DM incurs the processing cost and observes also $\sigma_2$, she gets an expected utility of
\begin{eqnarray*}
\mathbb{E}[u(\cdot)\mid \alpha_1 , \tilde\sigma_2] &=& \overline{U} \{
  P(\alpha_2 \mid  \alpha_1)p_{\alpha \alpha} + P(\beta_2 \mid 
  \alpha_1)[1- p_{\alpha \beta}] \}\\
 & & + \: \underline{U} \{
  P(\alpha_2 \mid  \alpha_1)[1- p_{\alpha \alpha}] + P(\beta_2 \mid 
  \alpha_1)p_{\alpha \beta}\} -c.
\end{eqnarray*}
We can proceed as before to calculate the DM's cost threshold by simplifying the following inequality,
\begin{eqnarray*}
c & \leq & \overline{U} \{
  P(\alpha_2 \mid  \alpha_1)p_{\alpha \alpha} + P(\beta_2 \mid 
  \alpha_1)[1- p_{\alpha \beta}] - [1- p_{\alpha}] \}  \\
 & & + \: \underline{U} [
  P(\alpha_2 \mid  \alpha_1)(1- p_{\alpha \alpha}) + P(\beta_2 \mid 
  \alpha_1)p_{\alpha \beta} - p_{\alpha}].
\end{eqnarray*}
Substituting (\ref{palfa}) for $p_{\alpha}$ gives that,
\begin{eqnarray*}
c & \leq & \overline{U} \{
  P(\alpha_2 \mid  \alpha_1)p_{\alpha \alpha} + P(\beta_2 \mid 
  \alpha_1)[1- p_{\alpha \beta}] - 1+ P(\alpha_2 \mid 
  \alpha_1) p_{\alpha \alpha}+ P(\beta_2 \mid  \alpha_1) p_{\alpha \beta} \}  \\ 
 & & + \: \underline{U} \{
  P(\alpha_2 \mid  \alpha_1)[1- p_{\alpha \alpha}] + P(\beta_2 \mid 
  \alpha_1) p_{\alpha \beta} -  P(\alpha_2 \mid 
  \alpha_1) p_{\alpha \alpha} - P(\beta_2 \mid  \alpha_1) p_{\alpha \beta} \} .
\end{eqnarray*} 
Further simplifications give
\begin{equation*}
c \leq \left[\overline{U} - \underline{U}\right] \cdot [P(\alpha_2 \mid  \alpha_1)] \cdot [2p_{\alpha \alpha} -1].
\end{equation*}
The set of admissible costs is $\mathcal{C}_{\alpha}(p)=\left[0,  \Delta U
  \cdot P(\alpha_2 \mid  \alpha_1) \cdot (2p_{\alpha \alpha} -1)
\right]$. By assumption, $\overline{U} - \underline{U}>0$,
$P(\alpha_2 \mid  \alpha_1)>0$, and $2p_{\alpha \alpha} -1>0$, as we
are considering the case $p_{\alpha \alpha}>\frac{1}{2}$. The
function $c_\alpha(p)= \Delta U \cdot P(\alpha_2 \mid  \alpha_1)
\cdot [2p_{\alpha \alpha} -1] $ is positive. Given $P(\alpha_2 \mid  \alpha_1)$ as in \eqref{eq:cond_alpha}, and $c_\alpha(p)  = \Delta U \cdot P(\alpha_2 \mid \alpha_1) \cdot [2p_{\alpha \alpha}-1]$, we have that
\begin{eqnarray*}
c_\alpha(p)& = & \Delta U \cdot
[\theta_2p_{\alpha} + (1-\theta_2)(1-p_{\alpha})]
\cdot [2p_{\alpha \alpha}-1] \\ 
 & =  &  \Delta U \cdot
\left[ \frac{(1-\theta_1 - \theta_2)(1-p) + \theta_1
    \theta_2}{\theta_1 p + (1- \theta_1)(1-p)} \right] \cdot \left[ \frac{\theta_1 \theta_2 p -
    (1-\theta_1)(1-\theta_2)(1-p)}{(1-\theta_1 - \theta_2)(1-p) + 
   \theta_1 \theta_2} \right].
\end{eqnarray*}
A DM that belongs to case 3 has the following cost function,
\begin{equation*}
c_\alpha(p)  = \Delta U \cdot \left[  \frac{\theta_1 \theta_2 p - (1-\theta_1)(1-\theta_2)(1-p)}{\theta_1 p + (1-\theta_1)(1-p)} \right]. 
\end{equation*}
The right hand side above is non-negative provided that $p \geq 1-\frac{\theta_1 \theta_2 }{1-\theta_1 -\theta_2 +2\theta_1 \theta_2}$.

\subsubsection*{Case 4: $p_{\alpha} \leq \frac{1}{2}$, $p_{\alpha \alpha} \leq
  \frac{1}{2}$, $p_{\alpha \beta} \leq \frac{1}{2}$}

In this case, we require the prior to be such that $p_{\alpha \alpha} \leq \frac{1}{2}$, implying that posteriors are always lower than $\nicefrac{1}{2}$. By using the explicit form for $p_{\alpha \alpha}$, we obtain $p \leq 1 - \frac{\theta_1 \theta_2 }{1- \theta_1 - \theta_2 + 2 \theta_1
  \theta_2}$. Recall from the analysis of case~3 that $1 - \frac{\theta_1
  \theta_2 }{1- \theta_1 - \theta_2 + 2 \theta_1 \theta_2} \leq 1-
  \theta_2$, and thus $p\leq 1-\theta_2$. This case is very similar to case~1: the DM's optimal choice is always $s^* = B$ regardless of the signals' realization. The DM's expected utility conditional on observing only $\sigma_1$ is
\begin{displaymath}
\mathbb{E}[u(\cdot)\mid \alpha_1]= (1 - p_{\alpha})\overline{U} + p_{\alpha}\underline{U}.
\end{displaymath}
Conditional on incurring the processing cost to observe also $\sigma_2$, the decision-maker has an expected payoff of
\begin{eqnarray*}
\mathbb{E}  [u(\cdot)\mid \alpha_1 , \tilde\sigma_2] &=& \overline{U} \{
  P(\alpha_2 \mid  \alpha_1)[1 - p_{\alpha \alpha}] + P(\beta_2 \mid 
  \alpha_1)[1 - p_{\alpha \beta}] \}  \\
 & & + \: \underline{U} [
  P(\alpha_2 \mid  \alpha_1)p_{\alpha \alpha} + P(\beta_2 \mid 
  \alpha_1) p_{\alpha \beta} ] -c.
\end{eqnarray*}
The DM is willing to incur the processing cost if and only if $\mathbb{E}[u(\cdot)\mid \alpha_1 , \tilde\sigma_2] \geq
\mathbb{E}[u(\cdot)\mid \alpha_1]$, which yields
\begin{eqnarray*}
c & \leq & \overline{U} \{
  P(\alpha_2 \mid  \alpha_1)[1- p_{\alpha \alpha}] + P(\beta_2 \mid 
  \alpha_1)[1- p_{\alpha \beta}] - [1-p_{\alpha}] \}  \\
 & & + \: \underline{U} \{
  P(\alpha_2 \mid  \alpha_1)p_{\alpha \alpha} + P(\beta_2 \mid 
  \alpha_1)p_{\alpha \beta} - p_{\alpha} \}.
\end{eqnarray*}
By substituting (\ref{palfa}) for $p_{\alpha}$, we obtain that
\begin{eqnarray*}
c & \leq & \overline{U} \{
  P(\alpha_2 \mid  \alpha_1)[1- p_{\alpha \alpha}] + P(\beta_2 \mid 
  \alpha_1)[1- p_{\alpha \beta}] - 1+ P(\alpha_2 \mid 
  \alpha_1) p_{\alpha \alpha}+ P(\beta_2 \mid  \alpha_1) p_{\alpha \beta}
  \} \\
 & & + \: \underline{U} \{
  P(\alpha_2 \mid  \alpha_1)p_{\alpha \alpha} + P(\beta_2 \mid 
  \alpha_1)p_{\alpha \beta} - [ P(\alpha_2 \mid  \alpha_1) p_{\alpha
    \alpha}+ P(\beta_2 \mid  \alpha_1) p_{\alpha \beta}] \} = 0.
\end{eqnarray*}
Similarly to Case 1, a decision-maker in case 4 is never willing to pay a positive processing cost $c$. The threshold for the processing cost, as indicated by the cost function, is $c_\alpha(p) = 0$, and zero is the only admissible cost, i.e., $\mathcal{C}_{\alpha}(p) = \{0\}$.

\subsubsection{Cases 5 to 8}

\subsubsection*{Case 5: $p_{\beta} \leq \frac{1}{2}$, $p_{\beta \alpha} \leq
  \frac{1}{2}$, $p_{\beta \beta} \leq \frac{1}{2}$}

The condition $p_{\beta \alpha} \leq \frac{1}{2}$ implies that $p_{\beta \beta} \leq \frac{1}{2}$. By using the explicit form of $p_{\beta \alpha}$ and setting it lower than or equal to $\nicefrac{1}{2}$, we obtain that $p \leq \frac{\theta_1 (1- \theta_2 )}{ \theta_1 + \theta_2 - 2 \theta_1
  \theta_2}$. We further require $p_{\beta} \leq \frac{1}{2}$, which implies $p<\theta_1$. This last condition is implied by $p_{\beta \alpha} \leq \frac{1}{2}$. The prior must satisfy
\begin{equation*}
p \in \left[0, \frac{\theta_1 (1- \theta_2 )}{ \theta_1 + \theta_2 - 2 \theta_1
  \theta_2}\right],
\end{equation*}
where $\frac{\theta_1 (1- \theta_2 )}{ \theta_1 + \theta_2 - 2 \theta_1
  \theta_2} \in \left[1- \theta_1 , \nicefrac{1}{2}\right]$ when $\theta_2 \leq
\frac{\theta_{1}^{2}}{1-2\theta_1 + 2\theta_{1}^{2}}$, and $\frac{\theta_1 (1- \theta_2 )}{ \theta_1 + \theta_2 - 2 \theta_1  \theta_2} \in [1- \theta_2 , 1-\theta_1]$ otherwise. The DM's expected utility
  conditional on observing only $\sigma_1=\beta$ is
\begin{displaymath}
\mathbb{E}[u(\cdot)\mid \beta_1]= (1 - p_{\beta})\overline{U} + p_{\beta}\underline{U}.
\end{displaymath}
Conditional on incurring the processing cost and observing $\tilde\sigma_2$, her expected payoff is
\begin{eqnarray*}
\mathbb{E}  [u(\cdot)\mid \beta_1 , \tilde\sigma_2] &=& \overline{U} \{
  P(\alpha_2 \mid  \beta_1)[1 - p_{\beta \alpha}] + P(\beta_2 \mid 
  \beta_1)[1 - p_{\beta \beta}] \} \\
 & & + \: \underline{U} [
  P(\alpha_2 \mid  \beta_1)p_{\beta \alpha} + P(\beta_2 \mid 
  \beta_1) p_{\beta \beta} ] -c.
\end{eqnarray*}
The DM is willing to pay the processing cost if and only if
$\mathbb{E}[u(\cdot)\mid \beta_1 , \tilde\sigma_2] \geq
\mathbb{E}[u(\cdot)\mid \beta_1)]$, which yields
\begin{eqnarray*}
c & \leq & \overline{U} \{
  P(\alpha_2 \mid  \beta_1)[1- p_{\beta \alpha}] + P(\beta_2 \mid 
  \beta_1)[1- p_{\beta \beta}] - [1-p_{\beta}] \} \\
 & & + \: \underline{U} \{
  P(\alpha_2 \mid  \beta_1)p_{\beta \alpha} + P(\beta_2 \mid 
  \beta_1)p_{\beta \beta} - p_{\beta} \}.
\end{eqnarray*}
Substituting (\ref{pbeta}) for $p_{\beta}$, we obtain that,
\begin{eqnarray*}
c & \leq & \overline{U} \{
  P(\alpha_2 \mid  \beta_1)[1- p_{\beta \alpha}] + P(\beta_2 \mid 
  \beta_1)[1- p_{\beta \beta}] - 1+ P(\alpha_2 \mid 
  \beta_1) p_{\beta \alpha}+ P(\beta_2 \mid  \beta_1) p_{\beta \beta}
  \}  \\
 & & + \: \underline{U} \{
  P(\alpha_2 \mid  \beta_1)p_{\beta \alpha} + P(\beta_2 \mid 
  \beta_1)p_{\beta \beta} - [ P(\alpha_2 \mid  \beta_1) p_{\beta
    \alpha}+ P(\beta_2 \mid  \beta_1) p_{\beta \beta}] \} = 0.
\end{eqnarray*}
A decision-maker belonging to case 4 is never willing to pay a positive processing cost $c$. The threshold for the processing cost, as indicated by the cost function, is $c_\beta(p) = 0$, and zero is the only admissible cost, i.e., $\mathcal{C}_{\beta}(p) = \{0\}$.

\subsubsection*{Case 6: $p_{\beta} \leq \frac{1}{2}$, $p_{\beta \alpha} \geq
  \frac{1}{2}$, $p_{\beta \beta} \leq \frac{1}{2}$}

The condition $p_{\beta} \leq \frac{1}{2}$ implies $p_{\beta  \beta} \leq \frac{1}{2}$. In addition, we need $p_{\beta \alpha} \geq
\frac{1}{2}$ satisfied. Those conditions together imply that the DM's prior must satisfy $p \leq \theta_1$ and $p \geq \frac{\theta_1 (1- \theta_2 )}{ \theta_1 + \theta_2 - 2 \theta_1   \theta_2}$. For this interval to exists, we need $\frac{\theta_1 (1- \theta_2 )}{ \theta_1 + \theta_2 - 2 \theta_1  \theta_2} \leq \theta_1$, which holds provided that $\theta_2 \geq
\frac{1}{2}$, thus always. A DM belongs to case 6 if her prior beliefs lie in the interval
\begin{equation*}
p \in \left[  \frac{\theta_1 (1- \theta_2 )}{ \theta_1 + \theta_2 - 2 \theta_1
  \theta_2} , \theta_1 \right].
\end{equation*}
The DM's expected utility conditional on observing only $\sigma_1=\beta$ is
\begin{displaymath}
\mathbb{E}[u(\cdot)\mid \beta_1]= (1 - p_{\beta})\overline{U} + p_{\beta}\underline{U}.
\end{displaymath}
Conditional on incurring the processing cost and observing $\tilde\sigma_2$, the DM has an expected payoff of
\begin{eqnarray*}
\mathbb{E}  [u(\cdot)\mid \beta_1 , \tilde\sigma_2] &=& \overline{U} \{
  P(\alpha_2 \mid  \beta_1)p_{\beta \alpha} + P(\beta_2 \mid 
  \beta_1)[1 - p_{\beta \beta}] \}  \\
 & & + \: \underline{U} [
  P(\alpha_2 \mid  \beta_1)[1- p_{\beta \alpha}] + P(\beta_2 \mid 
  \beta_1) p_{\beta \beta} ] -c.
\end{eqnarray*}
The DM is willing to pay the cost $c$ if and only if $\mathbb{E}[u(\cdot)\mid \beta_1 , \tilde\sigma_2] \geq
\mathbb{E}[u(\cdot)\mid \beta_1]$, that is, if and only if
\begin{eqnarray*}
c & \leq & \overline{U} \{
  P(\alpha_2 \mid  \beta_1)p_{\beta \alpha} + P(\beta_2 \mid 
  \beta_1)[1- p_{\beta \beta}] - [1-p_{\beta}] \}  \\
 & & + \: \underline{U} \{
  P(\alpha_2 \mid  \beta_1)[1- p_{\beta \alpha}] + P(\beta_2 \mid 
  \beta_1)p_{\beta \beta} - p_{\beta} \}.
\end{eqnarray*}
By substituting (\ref{pbeta}) for $p_{\beta}$, we obtain that,
\begin{displaymath}
c  \leq  \overline{U} [
  2P(\alpha_2 \mid  \beta_1)p_{\beta \alpha} + P(\beta_2 \mid 
  \beta_1)-1] - \underline{U} [
  2P(\alpha_2 \mid  \beta_1)p_{\beta \alpha} - P(\alpha_2 \mid 
  \beta_1) ]. 
\end{displaymath}
After a few simplifications and by using $P(\alpha_2\mid \beta_1)=1-P(\beta_2\mid \beta_1)$, we obtain
\begin{equation}
c \leq [\overline{U} - \underline{U}] \cdot [P(\alpha_2 \mid 
\beta_1)] \cdot [2p_{\beta \alpha} -1].
\end{equation}
The set of admissible cost is $\mathcal{C}_{\beta}(p)=\left[ 0 ,  \Delta U
  \cdot P(\alpha_2 \mid  \beta_1) \cdot [2p_{\beta \alpha} -1]\right]$. By assumption, $\overline{U} -\underline{U}>0$, $P(\alpha_2 \mid  \beta_1)>0$, and $2p_{\beta \alpha} -1>0$, as we are considering the case $p_{\beta \alpha}>\frac{1}{2}$. The threshold $c_\beta(p)= \Delta U \cdot P(\alpha_2 \mid  \beta_1) \cdot [2p_{\beta \alpha} -1] $ is positive. The cost function $c_\beta(p)$ depends on the payoff structure, the 
probability that the realization of $\tilde\sigma_2$ affects the
optimal choice $s^*$ (e.g., $\tilde\sigma_2=\alpha$ would change $s^*
= B$ to $s^* = A$) and on the posterior $p_{\beta \alpha}$. The threshold is increasing on the premium for a correct guess, on the probability of observing a signal that changes the optimal action, and on $p_{\beta \alpha}$. Given $P(\alpha_2 \mid  \beta_1)$ as in \eqref{eq:cond_alpha}, and
$c_\beta(p) = \Delta U \cdot [P(\alpha_2)] \cdot [2p_{\beta \alpha} -1]$, we have that
\begin{eqnarray*}
c_\beta(p) & = & \Delta U
\cdot \left[\frac{(1-\theta_1)\theta_2p +
  \theta_1(1-\theta_2)(1-p)}{(1-\theta_1)p + \theta_1(1-p)}\right] \cdot \left[\frac{(1-\theta_1)\theta_2p -
   \theta_1(1-\theta_2)(1-p)}{(1-\theta_1)\theta_2 p +
   \theta_1(1-\theta_2)(1-p)}\right]. 
\end{eqnarray*}
A DM with prior $p$ that belongs to case 6 has the following cost function,
\begin{equation*}
c_\beta(p) = \Delta U \cdot \left[ \frac{(1-\theta_1)\theta_2p -
   \theta_1(1-\theta_2)(1-p)}{(1-\theta_1)p + \theta_1(1-p)} \right].
\end{equation*}
The right hand side above is non-negative provided that $p \geq \frac{\theta_1 (1-\theta_2)}{(1-\theta_1)\theta_2 + \theta_1(1-\theta_2)}$.

\subsubsection*{Case 7: $p_{\beta} \geq \frac{1}{2}$, $p_{\beta \alpha} \geq
  \frac{1}{2}$, $p_{\beta \beta} \leq \frac{1}{2}$}

There are two necessary conditions on the posterior beliefs for a DM to be in case 7: $p_{\beta} \geq \frac{1}{2}$ and $p_{\beta \beta} \leq \frac{1}{2}$. The former implies $p_{\beta  \alpha} \geq \frac{1}{2}$ and $p \geq \theta_1$, while the latter implies $p \leq \frac{\theta_1 \theta_2 }{ 1 - \theta_1 - \theta_2 + 2 \theta_1  \theta_2}$. The threshold $\frac{\theta_1 \theta_2 }{ 1 - \theta_1 - \theta_2 + 2 \theta_1  \theta_2}$ is greater than $\theta_2$ provided that $\theta_2<1$, thus always. A DM belongs to case 7 when her prior belongs to the interval
\begin{equation}
p \in \left[ \theta_1 , \frac{\theta_1 \theta_2 }{ 1 - \theta_1 - \theta_2 + 2 \theta_1 \theta_2} \right].
\end{equation}
The DM's expected utility conditional on observing only $\sigma_1=\beta$ is
\begin{displaymath}
\mathbb{E}[u(\cdot)\mid \beta_1]= p_{\beta}\overline{U} + (1-p_{\beta})\underline{U}.
\end{displaymath}
Conditional on paying the processing cost and observing $\tilde\sigma_2$, the DM's expected payoff is
\begin{eqnarray*}
\mathbb{E}  [u(\cdot)\mid \beta_1 , \tilde\sigma_2] &=& \overline{U} \{
  P(\alpha_2 \mid  \beta_1)p_{\beta \alpha} + P(\beta_2 \mid 
  \beta_1)[1 - p_{\beta \beta}] \} \\
 & & + \: \underline{U} \{
  P(\alpha_2 \mid  \beta_1)[1- p_{\beta \alpha}] + P(\beta_2 \mid 
  \beta_1) p_{\beta \beta} \} -c.
\end{eqnarray*}
The DM is willing to pay the processing cost if and only if $\mathbb{E}[u(\cdot)\mid \beta_1 , \tilde\sigma_2] \geq
\mathbb{E}[u(\cdot)\mid \beta_1]$, that is, if and only if
\begin{eqnarray*}
c & \leq & \overline{U} \{
  P(\alpha_2 \mid  \beta_1)p_{\beta \alpha} + P(\beta_2 \mid 
  \beta_1)[1- p_{\beta \beta}] - p_{\beta} \}  \\
 & & + \: \underline{U} \{
  P(\alpha_2 \mid  \beta_1)[1- p_{\beta \alpha}] + P(\beta_2 \mid 
  \beta_1)p_{\beta \beta} - [1-p_{\beta}] \}. 
\end{eqnarray*}
By substituting (\ref{pbeta}) for $p_{\beta}$, we obtain 
\begin{displaymath}
c  \leq  \overline{U} \{
  P(\beta_2 \mid  \beta_1) - 2P(\beta_2 \mid 
  \beta_1)p_{\beta \beta}\} -  \underline{U} [1- P(\alpha_2\mid \beta_1) -
  2P(\beta_2 \mid  \beta_1)p_{\beta \beta} ]. 
\end{displaymath}
The condition boils down to
\begin{equation*}
c \leq \left[\overline{U} - \underline{U}\right] \cdot [P(\beta_2 \mid 
\beta_1)] \cdot [1- 2p_{\beta \beta}].
\end{equation*}
The set of admissible costs is $\mathcal{C}_\beta (p)= \left[ 0, \Delta U  \cdot P(\beta_2 \mid  \beta_1) \cdot [1- 2p_{\beta \beta}]
\right]$. By assumption, $\overline{U} - \underline{U}>0$, $P(\beta_2
\mid  \beta_1)>0$, and, since in this case we have $p_{\beta \beta}\leq\frac{1}{2}$, then
$1- 2p_{\beta \beta}$ is positive. The cost function $c_\beta(p)$ depends on the payoff structure, the probability that the realization of $\tilde \sigma_2$ affects $s^*$ (e.g., $\tilde \sigma_2= \beta$ would change $s^*=A$ to $s^*=B$) and finally on the posterior $p_{\beta \beta}$. The cost function is increasing on the premium
for a correct guess, on the probability of observing a signal that changes the optimal action and negatively on $p_{\beta
  \beta}$. Given $P(\beta_2 \mid  \beta_1)$ as in \eqref{eq:cond_beta}, and $c_\beta(p) \leq \Delta U \cdot [P(\beta_2\mid  \beta_1)] \cdot [1- 2p_{\beta \beta}]$, we have that
\begin{displaymath}
c_\beta(p) = \Delta U
\cdot \left[\frac{(1-\theta_1)(1-\theta_2)p +
  \theta_1 \theta_2(1-p)}{(1-\theta_1)p + \theta_1(1-p)}\right]
\cdot \left[\frac{\theta_1 \theta_2(1-p) -
   (1- \theta_1)(1-\theta_2)p}{(1-\theta_1)(1-\theta_2)p +  \theta_1 \theta_2(1-p)}\right].
\end{displaymath}
A DM with a prior $p$ that belongs to case 7 has the following cost function, 
\begin{equation*}
c_\beta(p) = \Delta U \cdot \left[ \frac{\theta_1 \theta_2(1-p) -
   (1- \theta_1)(1-\theta_2)p}{(1-\theta_1)p + \theta_1(1-p)} \right].
\end{equation*}
The right hand side above is non-negative provided that $p \leq \frac{\theta_1 \theta_2}{1-\theta_1 - \theta_2 + 2\theta_1 \theta_2}$.

\subsubsection*{Case 8: $p_{\beta} \geq \frac{1}{2}$, $p_{\beta \alpha} \geq
  \frac{1}{2}$, $p_{\beta \beta} \geq \frac{1}{2}$}

The condition $p_{\beta \beta} \geq \frac{1}{2}$ implies $p_{\beta} \geq \frac{1}{2}$ and $p_{\beta \alpha} \geq \frac{1}{2}$. Because $p_{\beta \beta} \geq \frac{1}{2}$ also implies $p \geq \frac{\theta_1 \theta_2 }{ 1 -
  \theta_1 - \theta_2   + 2 \theta_1 \theta_2} \geq \theta_2$ then, to be in case~8, the DM's prior must lie in the interval
\begin{equation*}
p \in \left[ \frac{\theta_1 \theta_2 }{ 1 - \theta_1 - \theta_2
  + 2 \theta_1 \theta_2}, 1 \right].
\end{equation*}
The DM's expected payoff conditional on observing only $\sigma_1=\beta$ is
\begin{displaymath}
\mathbb{E}[u(\cdot)\mid \beta_1]= p_{\beta}\overline{U} + (1-p_{\beta})\underline{U}.
\end{displaymath}
Conditional on incurring the processing cost and observing also $\tilde\sigma_2$, the DM receives an expected payoff of
\begin{eqnarray*}
\mathbb{E}  [u(\cdot)\mid \beta_1 , \tilde\sigma_2] &=& \overline{U} \{
  P(\alpha_2 \mid  \beta_1)p_{\beta \alpha} + P(\beta_2 \mid 
  \beta_1)p_{\beta \beta} \}  \\
 & & + \: \underline{U} [
  P(\alpha_2 \mid  \beta_1)(1-p_{\beta \alpha}) + P(\beta_2 \mid 
  \beta_1) (1-p_{\beta \beta}) ] -c.
\end{eqnarray*}
The DM is willing to pay the processing cost if and only if $\mathbb{E}[u(\cdot)\mid \beta_1, \tilde\sigma_2] \geq
\mathbb{E}[u(\cdot)\mid \beta_1]$, that is, if and only if
\begin{eqnarray*}
c & \leq & \overline{U} \{
  P(\alpha_2 \mid  \beta_1)p_{\beta \alpha} + P(\beta_2 \mid 
  \beta_1)p_{\beta \beta} -p_{\beta} \}  \\
 & & + \: \underline{U} \{
  P(\alpha_2 \mid  \beta_1)(1-p_{\beta \alpha}) + P(\beta_2 \mid 
  \beta_1)(1-p_{\beta \beta}) - [1-p_{\beta}] \}.
\end{eqnarray*}
By substituting (\ref{pbeta}) for $p_{\beta}$, we obtain 
\begin{eqnarray*}
c & \leq & \overline{U} \{
  P(\alpha_2 \mid  \beta_1)p_{\beta \alpha} + P(\beta_2 \mid 
  \beta_1)p_{\beta \beta} -P(\alpha_2 \mid  \beta_1)p_{\beta \alpha} - P(\beta_2 \mid 
  \beta_1)p_{\beta \beta} \} \\
 & & + \: \underline{U} \{
  P(\alpha_2 \mid  \beta_1)(1-p_{\beta \alpha}) + P(\beta_2 \mid 
  \beta_1)(1-p_{\beta \beta})  \\
 & & - \: 1 + P(\alpha_2 \mid  \beta_1)p_{\beta \alpha} + P(\beta_2 \mid 
  \beta_1)p_{\beta \beta}] \} = 0.
\end{eqnarray*}
A decision-maker in case 8 is never willing to pay a positive processing cost. The threshold for the processing cost, as indicated by the cost function, is $c_\beta(p) = 0$, and zero is the only admissible cost in this case, i.e., $\mathcal{C}_{\beta}(p) = \{0\}$.

\subsection{Proof of Proposition~\ref{prop:cost_properties}}\label{sec:costfunction}

\begin{proof}
Suppose that $\tilde\sigma_1=\alpha$. Previous results show that DMs with prior beliefs belonging to cases 1 and 4 are not willing to pay any positive processing cost for additional information. Conversely, DMs with prior beliefs belonging to case 2 and 3 have a cost functions that take strictly positive values, namely $c^2_{\alpha}(\cdot)=\Delta U \cdot \left[ \frac{\theta_2 (1-p) - \theta_1 p - \theta_1 \theta_2 (1-2p)}{\theta_1 p + (1-\theta_1)(1-p)} \right]$, and $c^3_{\alpha}(\cdot)=\Delta U \cdot \left[
  \frac{\theta_1 \theta_2 p - (1-\theta_1)(1-\theta_2)(1-p)}{\theta_1
    p + (1-\theta_1)(1-p)} \right]$, respectively. I construct the piecewise cost function conditional on $\tilde\sigma_1=\alpha$ as
\[ c_{\alpha}(p) = \left\{
 \begin{array}{l l}

 c^4_{\alpha}(p)= 0 & \quad \text{if $p \in  \left[ 0, 1 - \frac{\theta_1 \theta_2 }{1- \theta_1 - \theta_2 + 2 \theta_1
  \theta_2} \right]$},\\

 c^3_{\alpha}(p)=  \Delta U \cdot \left[
  \frac{\theta_1 \theta_2 p - (1-\theta_1)(1-\theta_2)(1-p)}{\theta_1 
    p + (1-\theta_1)(1-p)} \right] & \quad \text{if $p \in \left[  1 - \frac{\theta_1 \theta_2 }{1-
      \theta_1 - \theta_2 + 2 \theta_1 
  \theta_2} , 1-\theta_1 \right]$}, \\ 

  c^2_{\alpha}(p)=  \Delta U \cdot \left[ \frac{\theta_2 (1-p) - \theta_1 p -  
    \theta_1 \theta_2 (1-2p)}{\theta_1 p + (1-\theta_1)(1-p)} \right]
& \quad \text{if $p \in \left[  1-\theta_1 , \frac{\theta_2 (1- \theta_1)}{\theta_1 + \theta_2 - 2 \theta_1
  \theta_2} \right]$},\\

c^1_{\alpha}(p)= 0 & \quad \text{if $p \in \left[ \frac{ \theta_2 (1-
           \theta_1)}{\theta_1 + \theta_2 - 2 \theta_1 
  \theta_2} ,1 \right]$}.
\end{array}
\right.\]
The function $c_{\alpha}(p)$ is continuous in $p$. Indeed,
\begin{equation*}
\lim_{p \downarrow 1 - \frac{\theta_1 \theta_2 }{1- \theta_1 - \theta_2 + 2 \theta_1
  \theta_2}} c^{\alpha}_3(p) = \lim_{p \uparrow \frac{ \theta_2 (1-
           \theta_1)}{\theta_1 + \theta_2 - 2 \theta_1 
  \theta_2}} c^{\alpha}_2(p) = 0,
\end{equation*}
\begin{equation*}
\lim_{p \uparrow 1- \theta_1} c_{\alpha}^3(p) = \lim_{p \downarrow 1-
  \theta_1} c_{\alpha}^2(p) = \Delta U \left[ \theta_2 - \frac{1}{2} \right].
\end{equation*}
In addition, $c_\alpha^2(p)$ and $c_\alpha^3(p)$ are continuous in $p$. Moreover, $c^2_{\alpha}(p)$ is strictly decreasing in and convex in $p$, as the following inequalities hold true for all $p \in \left[  1-\theta_1 , \frac{\theta_2 (1-
    \theta_1)}{\theta_1 + \theta_2 - 2 \theta_1   \theta_2} \right]$,  
\begin{eqnarray*}
  \frac{\partial c^2_{\alpha}(p)}{\partial p} &=& -\Delta U \frac{\theta_1
  (1-\theta_1)}{[\theta_1(2p-1)-p+1]^2}  <0,  \\ 
\frac{\partial^2 c^2_{\alpha}(p)}{\partial p^2} &=&
\Delta U \frac{2\theta_1(3\theta_1 - 2\theta^2_1-1)}{[\theta_1(2p-1)-p+1]^3} >0.  
\end{eqnarray*}
By contrast, $c^3_{\alpha}(p)$ is strictly increasing and concave in $p$, as the following inequalities hold true for all $p \in \left[  1 - \frac{\theta_1 \theta_2 }{1-
      \theta_1 - \theta_2 + 2 \theta_1 
  \theta_2} , 1-\theta_1 \right]$,
\begin{eqnarray*}
\frac{\partial c^3_{\alpha}(p)}{\partial p} &=& \Delta U \frac{\theta_1
  (1-\theta_1)}{[\theta_1(2p-1)-p+1]^2}  >0,\\ 
\frac{\partial^2 c^3_{\alpha}(p)}{\partial p^2} &=&
\Delta U \frac{2\theta_1(2\theta^2_1 -3\theta_1 +1)}{[\theta_1(1-2p)+p-1]^3} <0.
\end{eqnarray*}

Thus, $c_{\alpha}(p)$ has a global maximum at $p=1-\theta_1$, where it takes value $\Delta U \left[ \theta_2 -  \frac{1}{2} \right]$. Moreover, the cost function is not differentiable at $p=1-\theta_1$, $1 -
\frac{\theta_1 \theta_2 }{1- \theta_1 - \theta_2 + 2 \theta_1  \theta_2}$, and $\frac{ \theta_2 (1- \theta_1)}{\theta_1 + \theta_2  - 2 \theta_1 \theta_2}$, because it has a kink in those points. While $c^1_{\alpha}(p)$ and
$c^4_{\alpha}(p)$ are constant at zero, $c^2_{\alpha}(p)$ and
$c^3_{\alpha}(p)$ have a null derivative 
only if $\theta_1=1$. Moreover, the first derivative of $c_{\alpha}(p)$'s right limit to $p\downarrow 1-\theta_1$ is the additive inverse of its limit from the left to $p\uparrow 1-\theta_1$, i.e., 
\begin{eqnarray*}
\lim_{p \downarrow 1-\theta_1} \frac{\partial c_{\alpha}(p)}{\partial
p} = \lim_{p \downarrow 1-\theta_1} \frac{\partial c^2_{\alpha}(p)}{\partial
p} &=& -\frac{\Delta U}{4\theta_1(1-\theta_1)}, \\
\lim_{p \uparrow 1-\theta_1} \frac{\partial c_{\alpha}(p)}{\partial
p} = \lim_{p \uparrow 1-\theta_1} \frac{\partial c^3_{\alpha}(p)}{\partial
p} &=& \frac{\Delta U}{4\theta_1(1-\theta_1)}.
\end{eqnarray*}

Suppose now that $\tilde\sigma_1=\beta$. Previous results show that DMs with prior beliefs belonging to cases 5 and 8 are not willing to pay any positive processing cost for additional information. Conversely, DMs with prior beliefs belonging to case 6 and 7 have a cost functions that take strictly positive values, namely $c^6_{\beta}(p)=\Delta U \cdot \left[
  \frac{(1-\theta_1)\theta_2p -
    \theta_1(1-\theta_2)(1-p)}{(1-\theta_1)p + \theta_1(1-p)} \right]$ and $c^7_{\beta}(p)=\Delta U \cdot \left[
  \frac{\theta_1 \theta_2(1-p) -
    (1-\theta_1)(1-\theta_2)p}{(1-\theta_1)p + \theta_1(1-p)}
\right]$, respectively.  I construct the piecewise cost function conditional on $\tilde\sigma_1=\beta$ as
\[ c_{\beta}(p) = \left\{
 \begin{array}{l l}

 c^5_{\beta}(p) = 0 & \quad \text{if $p \in  \left[ 0, \frac{\theta_1 (1- \theta_2 )}{ \theta_1 + \theta_2 - 2 \theta_1
  \theta_2} \right] $},\\

  c^6_{\beta}(p) = \Delta U \cdot \left[ \frac{(1-\theta_1)\theta_2p -
   \theta_1(1-\theta_2)(1-p)}{(1-\theta_1)p + \theta_1(1-p)} \right] &
\quad \text{if $p \in \left[ \frac{\theta_1 (1- \theta_2 )}{ \theta_1 + \theta_2 - 2 \theta_1
  \theta_2} , \theta_1 \right]$}, \\ 

   c^7_{\beta}(p) = \Delta U \cdot \left[ \frac{\theta_1 \theta_2(1-p) -
   (1- \theta_1)(1-\theta_2)p}{(1-\theta_1)p + \theta_1(1-p)} \right]
& \quad \text{if $p \in \left[ \theta_1 , \frac{\theta_1 \theta_2 }{ 1 - \theta_1 - \theta_2 + 2 \theta_1
  \theta_2} \right]$},\\

c^8_{\beta}(p) = 0 & \quad \text{if $p \in \left[ \frac{\theta_1 \theta_2 }{ 1 - \theta_1 - \theta_2 + 2 \theta_1
  \theta_2} ,1 \right]$}.
\end{array}
\right.\]
The function $c_{\beta}(p)$ is continuous
in $p$. Indeed, 
\begin{equation*}
\lim_{p \downarrow \frac{\theta_1 (1- \theta_2 )}{ \theta_1 + \theta_2 - 2 \theta_1
  \theta_2}} c_{\beta}^6(p) = \lim_{p \uparrow \frac{\theta_1 \theta_2 }{ 1 - \theta_1 - \theta_2 + 2 \theta_1
  \theta_2}} c_{\beta}^7(p) = 0,
\end{equation*}
\begin{equation*}
\lim_{p \uparrow \theta_1} c_{\beta}^6(p) = \lim_{p \downarrow  \theta_1}
c_{\beta}^7(p) = \Delta U \left[ \theta_2 - \frac{1}{2} \right]. 
\end{equation*}
In addition, $c_\beta^6(p)$ and $c_\beta^7(p)$ are continuous in $p$. Moreover, $c^6_{\beta}(p)$ is strictly increasing in and concave in $p$, as the following inequalities hold true for all $p \in \left[ \frac{\theta_1 (1- \theta_2 )}{
    \theta_1 + \theta_2 - 2 \theta_1 
  \theta_2} , \theta_1 \right]$,
\begin{eqnarray*}
\frac{\partial c^6_{\beta}(p)}{\partial p} &=& \Delta U \frac{\theta_1
  (1-\theta_1)}{[\theta_1(1-2p) + p]^2}  >0, \\ 
\frac{\partial^2 c^6_{\beta}(p)}{\partial p^2} &=&
\Delta U \frac{2\theta_1(2\theta^2_1 -3\theta_1 +1)}{[\theta_1(2p-1)-p]^3} <0.
\end{eqnarray*}
By contrast, $c^7_{\beta}(p)$ is strictly decreasing and convex in $p$, as the following inequalities hold true for all $p \in \left[ \theta_1 , \frac{\theta_1 \theta_2 }{ 1 - \theta_1 - \theta_2 + 2 \theta_1 \theta_2} \right]$,
\begin{eqnarray*}
\frac{\partial c^7_{\beta}(p)}{\partial p} &=& -\Delta U \frac{\theta_1
  (1-\theta_1)}{[\theta_1(1-2p) + p]^2}  <0, \\ 
\frac{\partial^2 c^7_{\beta}(p)}{\partial p^2} &=&
\Delta U \frac{2\theta_1(3\theta_1 -2\theta^2_1 -1)}{[\theta_1(2p-1)-p]^3} >0.
\end{eqnarray*}

Thus, $c_{\beta}(p)$ has a global maximum at $p=\theta_1$, where it takes value $\Delta U \left[ \theta_2 -  \frac{1}{2} \right]$. Moreover, the cost function is not differentiable at $p=\theta_1$, 
$\frac{\theta_1 (1- \theta_2 )}{\theta_1 + \theta_2 - 2 \theta_1 \theta_2}$, and $ \frac{\theta_1 \theta_2 }{1 - \theta_1 - \theta_2 + 2 \theta_1 \theta_2}$, because it has a kink in those points. While $c^5_{\beta}(p)$ and
$c^8_{\beta}(p)$ are constant at zero, $c^6_{\beta}(p)$ and
$c^7_{\beta}(p)$ have a null derivative 
only if $\theta_1=1$. Moreover, the first derivative of $c_{\beta}(p)$'s right limit to $p\downarrow \theta_1$ is the additive inverse of its limit from the left to $p\uparrow \theta_1$, i.e., 
\begin{eqnarray*}
\lim_{p \downarrow \theta_1} \frac{\partial c_{\beta}(p)}{\partial
p} = \lim_{p \downarrow \theta_1} \frac{\partial c^7_{\beta}(p)}{\partial
p} &=& -\frac{\Delta U}{4\theta_1(1-\theta_1)}, \\
\lim_{p \uparrow \theta_1} \frac{\partial c_{\beta}(p)}{\partial
p} = \lim_{p \uparrow \theta_1} \frac{\partial c^6_{\beta}(p)}{\partial
p} &=& \frac{\Delta U}{4\theta_1(1-\theta_1)}.
\end{eqnarray*}

The piece-wise cost functions $c_\alpha(p)$ and $c_\beta(p)$ are continuous because, as we have seen, all of their sub-functions are continuous on the corresponding subdomains, and there is no discontinuity at an endpoint of any subdomain. Moreover, it is easy to verify that $c_{\alpha}(p)=c_{\beta}(1-p)$ for all $p\in[0,1]$, meaning that the cost functions are symmetric with respect to $p$.
\end{proof}

\subsubsection{Proof of Corollary~\ref{cor:processing_priors}}\label{sec:inversecostfunctions}

\begin{proof}
From Proposition~\ref{prop:cost_properties}, it follows that there are two values of $p$ such that the cost function satisfies $c_{\sigma_1}(p)=c$ for some processing cost $c\in\left( 0, \Delta U \cdot [ \theta_2 - \frac{1}{2} ] \right)$. Moreover, we have that $c_{\sigma_1}(p)>c$ if and only if $p$ is in between the two thresholds at which $c_{\sigma_1}(p)=c$, thus defining the boundaries of the sets $\mathcal {H}_{\sigma_1}(c) \coloneqq \{ p\in [0,1] \mid 
c_{\sigma_1}( p) > c \}$. By setting $c_{\sigma_1}(p)=c$ and inverting the function $c_{\sigma_1}$ to solve for $p$, we find the two thresholds $\underline{\smash q}_{\sigma_1}(c)$ and $\overline{q}_{\sigma_1}(c)$, where $\underline{\smash q}_{\sigma_1}(c)<\overline{q}_{\sigma_1}(c)$. Through this process, we obtain
\begin{eqnarray}
\underline{\smash q}_{\alpha}(c) &=& \frac{(1- \theta_1)[\Delta U (\theta_2 - 1) -  c]}{\Delta U(\theta_1 + \theta_2 - 2 \theta_1 \theta_2 - 1) +  (2\theta_1 -1)c}, \nonumber \\
\overline{q}_{\alpha}(c) &=& \frac{(1- \theta_1)(\Delta U  \theta_2 -  c)}{\Delta U(\theta_1 + \theta_2 - 2 \theta_1 \theta_2 ) +  (2\theta_1 -1)c}, \nonumber \\
\underline{\smash q}_{\beta} (c) &=& \frac{\theta_1 [c + \Delta U (1- \theta_2)]}{  \Delta U(\theta_1 + \theta_2 - 2 \theta_1 \theta_2 ) +  (2\theta_1 -1)c}, \nonumber \\
\overline{q}_{\beta} (c)&=& \frac{\theta_1 (c - \Delta U \theta_2  )}{\Delta U(\theta_1 + \theta_2 - 2 \theta_1 \theta_2 - 1) +  (2\theta_1 -1)c}. \nonumber 
\end{eqnarray} 
As a result, we have that $\mathcal{H}_{\alpha}(c)=\left(\underline{\smash q}_{\alpha}(c) ,
\overline{q}_{\alpha}(c)\right)$ and
$\mathcal{H}_{\beta}(c)=\left(\underline{\smash q}_{\beta}(c) ,
\overline{q}_{\beta}(c)\right)$. 

Since $c_\alpha(p)$ and $c_\beta(p)$ reach a global maximum in $1-\theta_1$ and $\theta_1$, respectively, $c^3_{\alpha}(p)$ is monotone increasing and concave, $c^6_{\beta}(p)$ is monotone increasing and convex, then $c^3_{\alpha}(p)$ and $c^6_{\beta}(p)$ never cross in their shared domain. Likewise, $c^2_{\alpha}(p)$ and $c^7_{\beta}(p)$ are both monotonically decreasing, the former is convex and the latter is concave, and thus they do not cross within their shared domain. In fact, the only prior satisfying $c_{\alpha}(p,\cdot)=c_{\beta}(p,\cdot)$ is $p=\nicefrac{1}{2}$. It follows that $\underline{\smash q}_{\alpha}(c)<\underline{\smash q}_{\beta} (c)$ and $\overline{q}_{\alpha}(c) <\overline{q}_{\beta}(c)$. By the cost function's thresholds as in Proposition~\ref{prop:cost_properties}, we also obtain that $0<\underline{\smash q}_{\alpha}(c)<\overline{q}_{\beta}(c)<1$. Finally, by comparing $\underline{\smash q}_{\beta} (c)$ and $\overline{q}_{\alpha}(c)$ we find that $\underline{\smash q}_{\beta} (c)<\overline{q}_{\alpha}(c)$ if and only if $c\leq \Delta U \cdot [\theta_2-\theta_1]$.
\end{proof}

\subsection{Belief polarization}\label{app:polarization}

\begin{lemma}\label{lemma:necessary_PB}
Necessary conditions for PB to occur given some processing cost are
\begin{itemize}[noitemsep]
\item[i)] $(p^i,p^j) \in \mathcal{V}^{ij}_{\sigma_1}\cup  \mathcal{V}^{ji}_{\sigma_1}$ for some $\sigma_1\in\left\{\alpha_1,\beta_1\right\}$;
    \item[ii)] $\sigma_1 \neq \sigma_2$. 
\end{itemize}
Given a processing cost $c>0$, necessary conditions for PB are
    \begin{itemize}[noitemsep]
    \item[ii)] $\sigma_1 \neq \sigma_2$;
    \item[iii)] $c < \Delta U \cdot \left(\theta_2 - \nicefrac{1}{2}\right)$;
    \item[iv)] $\left(p^i,p^j\right)\in \mathcal{B}^{ij}_{\sigma_1}(c)\cup \mathcal{B}^{ji}_{\sigma_1}(c)$ for some $\sigma_1\in\left\{\alpha_1,\beta_1\right\}$.
    \end{itemize}  
\end{lemma}
\begin{proof}
Suppose by way of contradiction that $a^*_i(\sigma_1) = a^*_j(\sigma_1)$. Because DMs observe exactly the same information (either $\sigma_1$ or $(\sigma_1,\sigma_2)$), then they update their beliefs in the same direction, thus $I(\sigma)>0$. For PB to occur, it must be that $a^*_i(\sigma_1) \neq a^*_j(\sigma_1)$ for some $\sigma_1\in\left\{\alpha,\beta\right\}$. This is possible only if DMs have different willingness to pay for $\sigma_2$. From Definition~\ref{def:sets_sigma2}, it must be that condition $i)$ is satisfied. In such case, there are processing costs under which $a^*_i(\sigma_1) \neq a^*_j(\sigma_1)$. When $\sigma_1=\sigma_2$, the DMs update in the same direction regardless of weather they acquire $\tilde\sigma_2$ or not. Even in this case, $I(\sigma)>0$, meaning that there cannot be polarization of beliefs as per Definition~\ref{def:pb}. Hence, condition $ii)$.

Condition $iv)$ is analogous to condition $i)$ for the case where we consider a specific processing cost. It implies that, given $\sigma_1$ and $c>0$, one decision-maker will pay the processing cost to observe $\sigma_2$, while the other will not. Condition $iii)$ ensures that the set $\mathcal{H}_{\sigma_1}(c)$ is non-empty, while condition $ii)$ ensures that the inversion function is negative.
\end{proof}

Lemma~\ref{lemma:necessary_PB} says that, to study PB, we can focus the analysis on pairs of decision-makers whose prior beliefs belong to the set $\left\{ \mathcal{V}^{ij}_{\sigma_1} \cup \mathcal{V}^{ji}_{\sigma_1}  \right\}_{ \sigma_1\in\{\alpha,\beta\}}$. Those that hold extreme beliefs (i.e., $(p^i,p^j) \in \mathcal{E}^2$) never polarize because they never acquire the costly component regardless of $\sigma_1$. Those DMs that share the same prior beliefs never polarize because they make the same information acquisition choice. The next result focuses on those decision-makers whose priors fall into either $\mathcal{V}^{ij}_{\alpha}$ or $\mathcal{V}^{ji}_{\beta}$. It provides a condition on the information structure that is both necessary and also sufficient for probabilistic belief polarization.

\begin{lemma}\label{lemma:nec_suff_PB}
    Consider two decision-makers, $i$ and $j$, such that $p^i<p^j$. For every $\left(p^i,p^j\right) \in \mathcal{V}^{ij}_{\alpha}\cup \mathcal{V}^{ji}_{\beta}$, there exists a non-empty and convex set of processing costs such that, when $c$ belongs to that set, then $\theta_2>\theta_1$ is a necessary and sufficient condition to have PB with some positive ex-ante probability.
\end{lemma}

\begin{proof}
Consider first the case where $\sigma_1=\alpha$ and $(p^i,p^j) \in \mathcal{V}^{ij}_{\alpha}$. Lemma~\ref{lemma:necessary_PB} shows that necessary conditions for PB are $(p^i,p^j) \in \mathcal{V}^{ij}_{\sigma_1}\cup  \mathcal{V}^{ji}_{\sigma_1}$, $\sigma_2 = \beta$, and $a^*_i(\alpha) \neq a^*_j(\alpha)$. The former holds because $(p^i,p^j) \in \mathcal{V}^{ij}_{\alpha}$. The latter holds if only if $c \in \mathcal{C}_\alpha(p^i) \setminus \mathcal{C}_\alpha(p^j)$. The set $\mathcal{C}_\alpha(p^i) \setminus \mathcal{C}_\alpha(p^j)$ is non-empty because $(p^i,p^j) \in \mathcal{V}^{ij}_{\alpha}$ and $\mathcal{C}_\alpha(p^i) \supset \mathcal{C}_\alpha(p^j)$. It is convex because $\mathcal{C}_\alpha(p^k)$ is convex for $k\in\{i,j\}$. In this case, only DM~$i$ acquires and observes $\sigma_2$. The same argument holds when $\sigma_1=\beta$ and $(p^i,p^j) \in \mathcal{V}^{ji}_{\beta}$. In this case, only DM~$j$ acquires and observes $\sigma_2$ when $c \in \mathcal{C}_\beta(p^j) \setminus \mathcal{C}_\beta(p^i)$.

Recall that $\sigma_2\neq \sigma_1$ is necessary for PB (Lemma~\ref{lemma:necessary_PB}). Because $p_{\alpha}^j>p^j$, a necessary and sufficient condition to have $I(\alpha_1,\beta_2)<0$ is $p^i_{\alpha \beta}<p^i$. Such a condition is also sufficient for $D(\alpha_1,\beta_2)<0$, and it is satisfied if and only if $\theta_2>\theta_1$.  Because $p_\beta^i<p^i$, a necessary and sufficient condition to have $I(\beta_1,\alpha_2)<0$ is $p^j_{\beta \alpha}>p^j$. Such a condition is also sufficient for $D(\beta_1,\alpha_2)<0$, and it is satisfied if and only if $\theta_2>\theta_1$.

When all the necessary conditions are met (i.e., $(p^i,p^j) \in \mathcal{V}^{ij}_{\alpha}$ and $c \in \mathcal{C}_\alpha(p^i) \setminus \mathcal{C}_\alpha(p^j)$; or $(p^i,p^j) \in \mathcal{V}^{ji}_{\beta}$ and $c \in \mathcal{C}_\beta(p^j) \setminus \mathcal{C}_\beta(p^i)$), then PB occurs when $\sigma_2\neq \sigma_1$ if and only if $\theta_2>\theta_1$. Therefore, $\theta_2>\theta_1$ is, when the other conditions are satisfied, necessary and sufficient for PB to occur with positive probability, which is given by the ex-ante probability that $\sigma_2\neq \sigma_1$. 
\end{proof}

The next result addresses the two remaining sets of interest, $\mathcal{V}^{ij}_{\beta}$ and $\mathcal{V}^{ji}_{\alpha}$, completing the analysis of all possible cases that can lead to polarization of beliefs. In the following proposition, polarization of beliefs can only occur through a reversal in the order of beliefs, where the highest prior becomes the lowest posterior. The result is negative, indicating that these scenarios cannot lead to polarization of beliefs.

\begin{lemma}\label{lemma:PB_no_swap}
     Consider two decision-makers, $i$ and $j$, such that $p^i<p^j$. For all $(p^i,p^j) \in \mathcal{V}^{ji}_{\alpha}\cup\mathcal{V}^{ij}_{\beta}$ there can be inverse updating but not diverging attitudes, i.e., $D(\sigma)\geq 0$. In these cases, there cannot be PB.
\end{lemma}
\begin{proof}
    Suppose that $(p^i,p^j) \in \mathcal{V}^{ji}_{\alpha}$ and $\sigma_1=\alpha$. Lemma~\ref{lemma:necessary_PB} shows that $\sigma_2=\beta$ and $a^*_i(\alpha) \neq a^*_j(\alpha)$ are necessary to have PB. In this case, we have $c_\alpha(p^j)>c_\alpha(p^i)$, and thus to have PB it is necessary that $c\in \mathcal{C}_\alpha(p^j)\setminus \mathcal{C}_\alpha(p^i)$. Say that these necessary conditions are met. In this case, only decision-maker $j$ acquires and observes $\sigma_2$. By the definition of divergence function, $D(\sigma)<0$ requires $p^j - p^i < p^i_{\alpha} - p^j_{\alpha \beta}$. By substituting $p^j_{\alpha \beta}$ and rearranging, the condition becomes one about the relative precision of the signal's second component, that is,
    \begin{equation}\label{eq:conditiontheta2alpha}
    \theta_2 > \frac{\theta_1 p^j (1-p^i_{\alpha} + p^j - p^i)}{\theta_1
    p^j (1-p^i_{\alpha} + p^j - p^i) + (1-\theta_1)(1- p^j) (p^i_{\alpha} - p^j + p^i)}\eqqcolon g(p^i,p^j,\theta_1).
\end{equation}
We have that $g(p^i,p^j,\theta_1)\geq 1$ when $p^i_{\alpha} \leq p^j - p^i$. In such a case, there is no $\theta_2$ allowing for diverging attitudes, as $\theta_2 \in \left(\nicefrac{1}{2},1\right)$. PB requires that $p^i_{\alpha} > p^j - p^i$, meaning that DM~$i$'s posterior after observing $\sigma_1=\alpha$ should be sufficiently high. By substituting $p^i_{\alpha}$ and rearranging, we obtain a condition on the signal's first component precision, that is
\[
\theta_1 \underbrace{\left[ p^i - (p^j - p^i) (2p^i-1) \right]}_{>0 \text{ because } p^i<\nicefrac{1}{2}} > (p^j-p^i)(1-p^i).
\]
Before proceeding, it is useful to note that $p^i<\nicefrac{1}{2}$. From proposition~\ref{prop:cost_properties}, the highest willingness to pay for $\sigma_2$ after observing $\sigma_1=\alpha$, is at a prior belief equal to $1-\theta_1$. Since $c_\alpha(p^j)>c_\alpha(p^i)$ and $p^j>p^i$, it must be that $p^i<1-\theta_1<\nicefrac{1}{2}$. We can rewrite the previous inequality as
\[
\theta_1 > \frac{(p^j-p^i)(1-p^i)}{p^i + (p^j - p^i) (1-2p^i)}\eqqcolon h(p^i,p^j).
\]
We have that $h(p^i,p^j)>1$ when $p^j>p^i$, which is always true in the case considered here. There is no $\theta_1$ that allows for $p^i_{\alpha} > p^j - p^i$. As a result, it is not possible to have diverging attitudes when $(p^i,p^j) \in \mathcal{V}^{ji}_{\alpha}$ and $\sigma_1=\alpha$. 

A similar procedure shows that it is not possible to have diverging attitudes when $(p^i,p^j) \in \mathcal{V}^{ij}_{\beta}$ and $\sigma_1=\beta$. When $\sigma_2\neq \sigma_1$ and $a^*_i(\sigma_1) \neq a^*_j(\sigma_1)$, which occurs for suitable processing costs, there is inverse updating, i.e. $I(\sigma)<0$. Because PB requires both inverse updating and diverging attitudes, these cases never yield PB. 
\end{proof}

\begin{corollary}\label{cor:IFF_PB}
    Consider two decision-makers, $i$ and $j$, such that $p^i<p^j$, and a processing cost $c>0$. Then, PB occurs with positive ex-ante probability if and only if $(p^i,p^j)\in\mathcal{B}^{ij}_{\alpha}(c)\cup\mathcal{B}^{ji}_{\beta}(c)$, $c<\max\{c_{\sigma_1}(p^i),c_{\sigma_1}(p^j)\}$, and $\theta_2>\theta_1$.
\end{corollary}
\begin{proof}
    The proof draws from Lemmata \ref{lemma:nec_suff_PB} and \ref{lemma:PB_no_swap}. By Definition \ref{def:sets_sigma2}, if $(p^i,p^j)\in\mathcal{V}^{kl}_{\sigma_1}$, then it must be that $(p^i,p^j)\in\mathcal{B}^{kl}_{\sigma_1}(c)$ for some $c>0$, with $k,l\in\{i,j\}$, and $k\neq l$. Likewise, if $(p^i,p^j)\in\mathcal{B}^{kl}_{\sigma_1}(c)$ for a given $c>0$, and $\mathcal{B}^{kl}_{\sigma_1}(c)$ is non-empty, then $(p^i,p^j)\in\mathcal{V}^{kl}_{\sigma_1}$. As a result, the proofs of Lemmata \ref{lemma:nec_suff_PB} and \ref{lemma:PB_no_swap} readily apply to pairs of DMs such that $(p^i,p^j)\in\mathcal{B}^{kl}_{\sigma_1}(c)$. From this observation, it follows that PB can occur if and only if $(p^i,p^j)\in\mathcal{B}^{ij}_{\alpha}(c)\cup\mathcal{B}^{ji}_{\beta}(c)$ and $\theta_2>\theta_1$. The proof is completed by the observation that the set $\mathcal{B}^{kl}_{\sigma_1}(c)$ is non-empty if and only if $c<\max\{c_{\sigma_1}(p^i),c_{\sigma_1}(p^j)\}$.
\end{proof}

\subsubsection{Proof of Proposition~\ref{prop:PB_IFF}}
\begin{proof}
    The proof follows from Lemmata~\ref{lemma:necessary_PB}, \ref{lemma:nec_suff_PB}, \ref{lemma:PB_no_swap}, and Corollary~\ref{cor:IFF_PB}. Lemma~\ref{lemma:necessary_PB} outlines necessary conditions, the first one being $(p^i,p^j) \in \mathcal{V}^{ij}_{\sigma_1}\cup  \mathcal{V}^{ji}_{\sigma_1}$ for some $\sigma_1\in\left\{\alpha_1,\beta_1\right\}$. Lemma~\ref{lemma:nec_suff_PB} shows that for $\left(p^i,p^j\right) \in \mathcal{V}^{ij}_{\alpha}\cup \mathcal{V}^{ji}_{\beta}$ there are processing costs such that $\theta_2>\theta_1$ is necessary and sufficient for PB. Lemma~\ref{lemma:PB_no_swap} excludes that PB can occur when $(p^i,p^j) \in \mathcal{V}^{ji}_{\alpha}\cup\mathcal{V}^{ij}_{\beta}$. Together, these lemmata form necessary and sufficient conditions $i)$ and $ii)$. Corollary~\ref{cor:IFF_PB} concludes by setting necessary and sufficient conditions $iii)$ and $iv)$.
\end{proof}

\subsubsection{Proof of Proposition~\ref{prop:probability_PB}}
\begin{proof}
    Consider two decision-makers, $i$ and $j$, such that $p^i<p^j$. Suppose that conditions $i)$ to $iv)$ in Proposition \ref{prop:PB_IFF} are satisfied.\footnote{Recall that condition $iv)$ implies condition $i)$, meaning that we only need conditions $ii)$ to $iv)$ to hold.} Importantly, polarization of beliefs occurs with strictly positive ex-ante probability only if $(p^i,p^j)\in\mathcal{B}^{ij}_{\alpha}(c)\cup\mathcal{B}^{ji}_{\beta}(c)$. Lemma~\ref{lemma:necessary_PB} shows that $\sigma_1\neq\sigma_2$ is necessary for polarization to occur. We can write the ex-ante probability of PB as
    \begin{multline}
    Pr(PB)=Pr(\tilde\sigma=(\alpha,\beta))\cdot \mathds{1}\left\{ (p^i,p^j) \in\mathcal{B}^{ij}_{\alpha}(c) \right\}\\
    + Pr(\tilde\sigma=(\beta,\alpha))\cdot \mathds{1}\left\{ (p^i,p^j) \in\mathcal{B}^{ji}_{\beta}(c)\right\},
    \end{multline}
    where $Pr(\tilde\sigma=(\sigma_1,\sigma_2))=Pr(\tilde\sigma_1 = \sigma_1)\cdot Pr(\tilde\sigma_2 = \sigma_2)$, and
    \[
    Pr(\tilde\sigma_j=\alpha)= p\theta_j + (1-p)(1-\theta_j),
    \]
    \[
    Pr(\tilde\sigma_j=\beta)= p(1-\theta_j) + (1-p)\theta_j.
    \]
    The proposition's probability $Pr(PB)$ is computed from \eqref{eq:probPB} by using a subjective probability that the state is $s=A$ equal to $p$. The probability is maximized when the conditions in the indicator functions are both satisfied. The set $\mathcal{B}^{ij}_{\alpha}(c) \cap \mathcal{B}^{ji}_{\beta}(c)$ can be non-empty, as can be graphically inferred from Figure~\ref{fig:setsJK}. Consider the case where $(p^i,p^j) \in \mathcal{B}^{ij}_{\alpha}(c) \cap \mathcal{B}^{ji}_{\beta}(c)\neq \varnothing$. In this case, the ex-ante probability of PB boils down to
    \begin{multline*}
    Pr(PB)|_{(p^i,p^j)\in\mathcal{B}^{ij}_{\alpha}(c)\cup\mathcal{B}^{ji}_{\beta}(c)}
    = Pr(\alpha_1)Pr(\beta_2)+ Pr(\beta_1)Pr(\alpha_2)\\
    = \left( \theta_1+\theta_2-2\theta_1\theta_2 \right)\left[ 1-4p(1-p) \right] + 2p(1-p).
    \end{multline*}
    The component $\left[ 1-4p(1-p) \right]$ is non-negative for all $p\in[0,1]$ and equal to zero for $p=\nicefrac{1}{2}$. The last part $2p(1-p)$ is maximized for $p=\nicefrac{1}{2}$. The first component, which amounts to $\left( \theta_1+\theta_2-2\theta_1\theta_2 \right)$, is maximized for $\theta_1\to\nicefrac{1}{2}^+$ and $\theta_2\to 1^-$, yielding
    \[
    \lim_{\substack{%
    \theta_1\to\nicefrac{1}{2}^+\\
    \theta_2\to 1^-}} \left( \theta_1 + \theta_2 -2\theta_1\theta_2 \right) = \nicefrac{1}{2}.
    \]
    Furthermore, we have that
    \[
    \frac{\partial Pr(PB)|_{(p^i,p^j)\in\mathcal{B}^{ij}_{\alpha}(c)\cup\mathcal{B}^{ji}_{\beta}(c)}}{\partial p} = -2(2p-1)(2\theta_1-1)(2\theta_2-1),
    \]
    and thus such a probability is maximized when either $p=\nicefrac{1}{2}$ or when $\theta_1\to\nicefrac{1}{2}^+$.
\end{proof}

\subsubsection{Proof of Corollary~\ref{cor:polarizing}}
\begin{proof}
    Suppose that $\sigma_1=\alpha$. Consider a processing cost $c \in \left(0, \Delta U \left[\theta_2 - \frac{1}{2}\right] \right)$, and the sets $\mathcal{H}_{\alpha}(c)$, $\left[0, \underline{\smash q}_\alpha (c)\right]$, and $\left[\overline{q}_\alpha (c),1\right]$ (see Definition~\ref{def:hulsets} and Corollary~\ref{cor:processing_priors}). By definition, these sets are non-empty for some $c$, with $\left[0, \underline{\smash q}_\alpha (c)\right]\cup \mathcal{H}_{\alpha}(c) \cup \left[\overline{q}_\alpha (c),1\right]=[0,1]$. Moreover, $\left[0, \underline{\smash q}_\alpha (c)\right]\cap \mathcal{H}_{\alpha}(c)=\varnothing$, $\mathcal{H}_{\alpha}(c) \cap \left[\overline{q}_\alpha (c),1\right]=\varnothing$, and $\left[0, \underline{\smash q}_\alpha (c)\right]\cap\left[\overline{q}_\alpha (c),1\right]=\varnothing$. From Lemma~\ref{lemma:nec_suff_PB}, if $p^i \in \mathcal{H}_{\alpha}(c)$, then $\theta_2>\theta_1$ is sufficient to have PB with positive ex-ante probability between $i$ and $j$ for any DM~$j$ such that $p^j \in \left[\overline{q}_\alpha (c),1\right]$. That would imply $(p^i,p^j)\in \mathcal{V}^{ij}_{\alpha}$ and $p^i<p^j$. In this case, we have $\mathcal{T}\supseteq\left[\overline{q}_\alpha (c),1\right]$. Likewise, if $p^i \in \left[\overline{q}_\alpha (c),1\right]$, then we have $\mathcal{T}\supseteq\mathcal{H}_{\alpha}(c)$. Finally, suppose that $p^i\in \left[0, \underline{\smash q}_\alpha (c)\right]$. Corollary~\ref{cor:processing_priors} shows that $\underline{\smash q}_\alpha (c)<\underline{\smash q}_\beta (c)$. From Lemma~\ref{lemma:nec_suff_PB}, if $p^i \in \left[0, \underline{\smash q}_\alpha (c)\right]\subset\left[0, \underline{\smash q}_\beta (c)\right]$, then $\theta_2>\theta_1$ is sufficient to have PB with positive ex-ante probability between $i$ and $j$ for any DM~$j$ such that $p^j \in\mathcal{H}_{\beta}(c)$. In this last case, $\mathcal{T}\supseteq\mathcal{H}_{\beta}(c)$. A similar argument completes the proof for the case where $\sigma_1=\beta$.
\end{proof}


\subsection{Confirmatory belief patterns and reaction to information}\label{app:patterns}

\subsubsection{Proof of Proposition~\ref{prop:disconfirmation}}
\begin{proof}
    The proof follows directly from Proposition \ref{prop:cost_properties}, which analyzes the decision-maker's willingness to pay for additional information. The proposition shows that when $p>\nicefrac{1}{2}$ and $p\in\neg\mathcal{E}_\beta$, then $c_\beta(p)>c_\alpha(p)\geq 0$. Likewise, when $p<\nicefrac{1}{2}$ and $p\in\neg\mathcal{E}_\alpha$, then $c_\alpha(p)>c_\beta(p)\geq 0$.  The cost function satisfies $c_\beta(p)>0$ for all $p\in\neg\mathcal{E}_\beta$, and $c_\beta(p)=0$ otherwise. Likewise, $c_\alpha(p)>0$ for all $p\in\neg\mathcal{E}_\alpha$, and $c_\alpha(p)=0$ otherwise. The sets $\neg\mathcal{E}_{\sigma_1}$ are non-empty and convex, and $\neg\mathcal{E}=\neg\mathcal{E}_\alpha \cup \neg\mathcal{E}_\beta$ by Definition \ref{def:extremeset}. Moreover, $\inf\neg \mathcal{E}_\alpha<\frac{1}{2}<\sup\neg \mathcal{E}_\beta$, meaning that there are always prior beliefs that yield to a tendency for disconfirmation. It follows that every DM with prior $p\in \neg\mathcal{E}\setminus\left\{\frac{1}{2}\right\}$ has a tendency for disconfirmation.

    Consider a DM with prior $p>\nicefrac{1}{2}$. Such a DM exhibits disconfirmation when $a^*(\beta_1)=\rho \neq a^*(\alpha_1)=\neg\rho$. This is possible if and only if the processing cost is such that $c\in\left( c_\alpha(p),c_\beta(p) \right)$. For a prior $p>\nicefrac{1}{2}$, we have that $c_\beta(p)>c_\alpha(p) \geq 0$ if and only if $p\in\neg\mathcal{E}_\beta$. The same line of reasoning applies to $p<\nicefrac{1}{2}$ and $p\in\neg\mathcal{E}_\alpha$. The sets $\neg \mathcal{E}_\beta$ and $\neg \mathcal{E}_\alpha$ are disjoint for $\theta_1\geq\theta_2$, and overlap when $\theta_2>\theta_1$, with $\neg \mathcal{E}_\alpha\cap \neg \mathcal{E}_\beta=\left(\frac{\theta_1 (1-\theta_2)}{\theta_1 + \theta_2 - 2 \theta_1 \theta_2},\frac{ \theta_2 (1-\theta_1)}{\theta_1 + \theta_2 - 2\theta_1 \theta_2} \right)$. For prior beliefs in the set $\neg \mathcal{E}_\alpha\cap \neg \mathcal{E}_\beta$, we have that $\min\{c_\alpha(p),c_\beta(p)\}>0$. In those cases, disconfirmation occurs if and only if $ c\in \left( \min\{ c_\alpha(p),c_\beta(p) \}, \max\{ c_\alpha(p),c_\beta(p) \} \right)$.  The prior $p=\nicefrac{1}{2}$ is removed because it does not favor any state, and $c_\alpha(\nicefrac{1}{2})=c_\beta(\nicefrac{1}{2})$.
\end{proof}

\subsubsection{Proof of Proposition~\ref{prop:confirmation}}

\begin{proof}
Consider a decision-maker with prior $p>\nicefrac{1}{2}$. By definition, CB occur when ${p}(\sigma)>p>{p}_{\sigma}$. For ${p}(\sigma)\neq {p}_{\sigma}$, it must be that $a^*(\sigma_1)=\neg \rho$. Otherwise, $a=\rho$ implies that ${p}(\sigma)={p}_{\sigma}$.  For ${p}(\sigma)>p>{p}_{\sigma}$ it must be that $\sigma=(\alpha,\beta)$ and that $\theta_2>\theta_1$. Otherwise, given that $a=\neg\rho$, a signal $\sigma=(\beta,\beta)$ would yield $p_{\sigma}<p(\sigma)<p$, a signal $\sigma=(\alpha,\alpha)$ would yield $p_{\sigma}>p(\sigma)>p$, and a signal $\sigma=(\beta,\alpha)$ would yield ${p}(\sigma)<p<{p}_{\sigma}$. Likewise, given that $\sigma=(\alpha,\beta)$, $\theta_2\leq\theta_1$ would imply $p\leq p_{\sigma}<p(\sigma)$. The optimal choice $a^*(\alpha_1)=\neg \rho$ can occur only when $c>c_{\alpha}(p)$. To sum up, necessary conditions for CB in this case are $\theta_2>\theta_1$, $c>c_{\alpha}(p)$, and $\sigma=(\alpha,\beta)$. These conditions are also sufficient as they imply ${p}(\sigma)>p>{p}_{\sigma}$. The same line of reasoning applies to DB and to the case decision-makers with prior $p<\nicefrac{1}{2}$, completing the proof.
\end{proof}

\subsubsection{Proof of Proposition~\ref{prop:reaction}}

\begin{proof}
    It is necessary for both UR and OR that $a^*(\sigma_1) = \neg\rho$, because otherwise $p_{\sigma} = p(\sigma)$. We have $a^*(\sigma_1) = \neg\rho$ if and only if $c > c_{\sigma_1}(p)$. UR requires $\sigma_1 = \sigma_2$; otherwise, if $\sigma_1 \neq \sigma_2$, then neither $p < p(\sigma) < p_{\sigma}$ nor $p_{\sigma} < p(\sigma) < p$ can hold. OR requires $\sigma_1 \neq \sigma_2$ and $\theta_2 < \theta_1$; otherwise, if $\sigma_1 = \sigma_2$, then neither $p < p_{\sigma} < p(\sigma)$ nor $p(\sigma) < p_{\sigma} < p$ can hold. The same applies when $\theta_2 \geq \theta_1$, as in that case beliefs satisfy $p_{\alpha\beta} \leq p$ and $p_{\beta\alpha} \geq p$. These conditions are also sufficient. If $c > c_{\sigma_1}(p)$ and $\sigma_1 = \sigma_2$, then we obtain $p < p(\sigma) < p_{\sigma}$ when $\sigma = (\alpha, \alpha)$ and $p_{\sigma} < p(\sigma) < p$ when $\sigma = (\beta, \beta)$. If $c > c_{\sigma_1}(p)$, $\sigma_1 \neq \sigma_2$, and $\theta_2 < \theta_1$, then we obtain $p < p_{\sigma} < p(\sigma)$ when $\sigma = (\alpha, \beta)$ and $p(\sigma) < p_{\sigma} < p$ when $\sigma = (\beta, \alpha)$.
\end{proof}

\bibliography{biblio_beliefs} 
\end{document}